\documentclass[acmsmall,nonacm]{acmart/acmart}

\usepackage{graphicx} 
\usepackage{amsmath}
\usepackage{graphicx}
\usepackage{float}
\usepackage{xcolor}
\usepackage{todonotes}
\usepackage{doi}
\usepackage[capitalize]{cleveref}
\usepackage{textcomp}
\usepackage{thm-restate}
\usepackage{enumitem}

\usepackage{microtype} 
\usepackage{xurl}         
\input{sty/macros}

\newcommand{\vin}{v_{in}}

\newtheorem{property}{Property}

\usepackage{algorithm}
\makeatletter\newcommand{\algmargin}{\the\ALG@thistlm}
\algdef{SE}[UPON]{Upon}{EndUpon}[3]{\textbf{upon} \textbf{#1} $\langle \text{#2} \mid #3\rangle$ \textbf{do}}{}
\algdef{SE}[UPONSIMPLE]{UponSimple}{EndUpon}[1]{\textbf{upon} #1 \textbf{do}}{}
\algdef{SE}[InParallel]{InParallel}{EndParallel}[1]{\textbf{in parallel} #1 \textbf{do}}{}
\algdef{SE}[UPONEXISTS]{UponExists}{EndUpon}[1]{\textbf{upon exists} #1 \textbf{do}}{}
\algdef{SE}[UPONTRUE]{UponTrue}{EndUpon}[1]{\textbf{upon once} #1 \textbf{do}}{}
\algdef{SE}[USES]{Uses}{EndUses}{\textbf{Uses:}}{}
\ifthenelse{\equal{\ALG@noend}{t}}%
  {\algtext*{EndUses}}
  {}%
\ifthenelse{\equal{\ALG@noend}{t}}%
  {\algtext*{EndUpon}}
  {}%
\algtext*{EndUpon}
\makeatother
\algnewcommand{\parState}[1]{
    \parbox[t]{\dimexpr\linewidth-\algmargin}{\strut\hangindent=\algorithmicindent \hangafter=1 #1\strut}}
\algblock{As}{EndAs}
\algnewcommand\algorithmicas{\textbf{as}}
\algrenewtext{As}[1]{\algorithmicas\ #1}
\algtext*{EndAs}



\author{Andrei Constantinescu}
\affiliation{%
  \institution{Category Labs \& ETH Zurich}
  \country{Switzerland}
}
\email{aconstantine@ethz.ch}

\author{Marc Dufay}
\orcid{0009-0005-8440-8007}
\affiliation{%
  \institution{ETH Zurich}
   \country{Switzerland}
}
\email{mdufay@ethz.ch}

\author{Anton Paramonov}
\orcid{0009-0000-0760-8746}
\affiliation{%
  \institution{ETH Zurich}
   \country{Switzerland}
}
\email{aparamonov@ethz.ch}

\author{Roger Wattenhofer}
\orcid{0000-0002-6339-3134}
\affiliation{%
  \institution{ETH Zurich}
   \country{Switzerland}
}
\email{wattenhofer@ethz.ch}

\begin{document}
\title{From Few to Many Faults: Optimal Adaptive Byzantine Agreement}

\begin{abstract}
    Achieving agreement among distributed parties is a fundamental task in modern systems, underpinning applications such as consensus in blockchains, coordination in cloud infrastructure, and fault tolerance in critical services. However, this task can be intensive, often requiring a large number of messages to be exchanged as well as many rounds of communications, especially in the presence of Byzantine faults. This makes efficiency a central challenge in the design of practical agreement protocols.
    
    In this paper, we study the problem of Binary Agreement and give protocols that are simultaneously optimal in both message and round complexity, parameterized by the \emph{actual} number of Byzantine faults. In contrast to previous works, we demonstrate that optimal message complexity can be achieved without sacrificing latency. Concretely, for a system of \(n\) parties tolerating up to \(t\) Byzantine faults, out of which only \(f \leq t\) are actually faulty, we give the following results:

    \begin{itemize}
        \item When $t = \Omega(n)$, in the synchronous (resp. partially synchronous) setting, with optimal resiliency $t < n/2$ (resp. $t < n/3$), we describe a deterministic protocol with optimal communication complexity $\mathcal{O}(n \cdot (f+1))$ and optimal round complexity $\mathcal{O}(f + 1)$.
        \item Building upon this previous result, when $t = o(n)$, for both the synchronous and partially synchronous setting, we describe a deterministic protocol with near-optimal communication complexity $\widetilde{\mathcal{O}}(n + t\cdot f)$ and near-optimal round complexity $\widetilde{\mathcal{O}}(f+1)$. Our approach relies on a novel use of dispersers to efficiently disseminate a value.
        \item For the asynchronous setting, we show a $\Omega(n + t^2)$ lower bound in expectation and provide a randomized protocol with near-optimal $\widetilde{\mathcal{O}}(n + t^2)$ communication complexity and $\mathcal{O}(1)$ round complexity in expectation.
    \end{itemize}
    
    
\end{abstract}

\maketitle

\newpage
\pagenumbering{arabic}  
\setcounter{page}{1}

\section{Introduction}

Achieving agreement in a distributed setting is fundamental to blockchain systems. Modern blockchains often comprise thousands of nodes attempting to reach consensus \cite{Ethernodes2025, SolanaValidators2025, CardanoStakePools2025, AvalancheValidators2025}. Under the widely accepted resilience model in which up to $t \approx n/3$ validators may be Byzantine, and in light of the classical Dolev--Reischuk lower bound of $\Omega(t^2)$ messages for agreement \cite{PODC:DolRei82}, such systems must exchange millions of messages. This communication burden translates into high latency and limited scalability.

This work tackles the problem through two key contributions.

First, we design \emph{adaptive}\footnote{Not to be confused with algorithms tolerating an adaptive adversary.} algorithms, i.e., algorithms whose performance improves as the actual number of faults in a given execution, $f$, decreases (while $t$ naturally stays fixed). In particular, we present the first synchronous and partially synchronous Byzantine Agreement algorithms that simultaneously achieve adaptive communication $\mathcal{O}(nf)$ and adaptive latency $\mathcal{O}(f)$,\footnote{We note that, when we write $\mathcal{O}(nf)$ or $\mathcal{O}(f)$ we do not imply that with $f = 0$ our algorithm results in $0$ communication and $0$ rounds. Not to overload notation with $(f+1)$-s everywhere, we will assume that $f$ is at least $1$ in all the complexities.} while maintaining optimal resilience. This is significantly lighter than the worst-case $\Omega(t^2) = \Omega(n^2)$ communication suggested by \cite{PODC:DolRei82}, and is especially relevant in practice, where the number of misbehaving nodes is almost always negligibly small.

Second, we argue that systems' robustness is governed not by the relative \emph{fraction} of corruptions tolerated, but by the \emph{absolute} number of failures a system can withstand. For example, it may be easier for an adversary to corrupt a system with $3{,}001$ nodes that tolerates $1{,}000$ faults than one with $10{,}000$ nodes tolerating $2{,}000$ faults. Motivated by this perspective, we design an algorithm that decouples $n$ and $t$, yielding improved performance when $t \ll n$ and enabling unprecedented scaling in the total number of nodes. To get a taste of it, imagine a protocol achieving $\mathcal{O}(n + t^2)$ communication. Using it, system designers can first choose a failure threshold $t$ such that disseminating $t^2$ messages remains feasible, and subsequently scale $n$ up to $t^2$ without incurring additional asymptotic communication overhead.

Combining these two insights, we circumvent the Dolev--Reischuk lower bound and obtain the first synchronous and partially synchronous Byzantine Agreement algorithms with near-optimal communication complexity $\widetilde{\mathcal{O}}(n + t \cdot f)$ and round complexity $\widetilde{\mathcal{O}}(f)$, where $\widetilde{\mathcal{O}}$ hides polylogarithmic factors.

We also contribute to asynchronous agreement. Whereas existing approaches typically incur $\mathcal{O}(n \cdot t)$ communication, our asynchronous algorithm achieves communication complexity $\widetilde{\mathcal{O}}(n + t^2)$.

Finally, we show that adaptive communication is impossible in the asynchronous setting. Specifically, we prove an $\Omega(n + t^2)$ lower bound for any algorithm that guarantees almost-sure termination. This result implies that our asynchronous protocol is optimal up to logarithmic factors.

\section{Related Work}

\paragraph*{Synchrony.} Optimizing the communication complexity of Byzantine Agreement (BA) has been the focus of extensive research, particularly in synchronous networks.

Cohen et al. \cite{PODC:CohKeiSpi22} present a deterministic solution for Weak Byzantine Agreement and Byzantine Broadcast (BB) in synchronous settings, achieving a communication complexity of $\mathcal{O}(n\cdot f)$. Consequently, Elsheimy et al. \cite{SODA:ElsTsiPap24} and Civit et al. \cite{civit2023strong} independently developed deterministic protocols for Byzantine Agreement with $\mathcal{O}(n \cdot f)$ communication complexity, maintaining resilience $t < (\frac{1}{2} - \varepsilon)n$. 

This line of research culminated in the work of Civit et al. \cite{PODC:CDGGKV24}, which extends to Multi-Valued Byzantine Agreement (MVBA), where processes can propose arbitrary values and not just $0$ or $1$, and Interactive Consistency which ensures that the output is a vector of proposals from honest processes, while also supporting External Validity, meaning the decided value must satisfy a predetermined predicate. The bit communication complexity of this protocol is $\mathcal{O}(L_o n + n(f + 1)\kappa)$, where $L_o$ is the bit-length of the output (e.g., $L_o = 1$ for Binary BA), and $\kappa$ is a security parameter. Notably, \cite{PODC:CDGGKV24} uses a heavy cryptography primitive called Multiverse Threshold Signature Scheme (MTSS) \cite{baird2023threshold,garg2024hints} and poses the question of whether adaptive BA can be implemented with lighter cryptography, which the present work answers affirmatively.  

All aforementioned protocols have a round complexity of $\Omega(n)$ in the worst case, even if $f$ is small. Another area of research is early stopping protocols, like the one by Lenzen and Sheikholeslami \cite{lenzen2022recursive}, which solves BA with round complexity $\mathcal{O}(f)$, having, however, $\mathcal{O}(n \cdot t)$ communication complexity and $t < n/3$ resilience. We remark that early stopping protocols focus on the stopping time, i.e the time for all honest parties to exit the protocol. This is a stronger concept compared to this paper, where round complexity characterizes the time for all honest parties to decide a value, because a party may still send messages after deciding. Independently and in parallel with our work, Blum, Lenzen and Loss gave a randomized early-stopping protocol with $\mathcal{O}(n \cdot f)$ round complexity \cite{blum2026earlyadaptive}. 

On the impossibility front, the classical lower bound by Dolev and Reischuk \cite{PODC:DolRei82} establishes that any synchronous algorithm solving Byzantine Agreement must incur a message complexity of $\Omega(f^2)$. This foundational result was further refined by Spiegelman \cite{spiegelman2020search}, who demonstrated a lower bound of $\Omega(n + tf)$. Expanding on these results, Abraham et al. \cite{PODC:ACDNPR19} proved that even randomized algorithms are constrained by the $\Omega(f^2)$ message complexity in synchronous settings when facing a highly adaptive adversary, one capable of performing after-the-fact removal of messages, i.e when an adversary can see which messages are about to be sent and use this information to decide which parties to corrupt and modify their messages.

As for round complexity, the classical result by Dolev, Reischuk and Strong \cite{dolev1990early} establishes the lower bound of $\min\{f + 2, t + 1\}$ rounds for achieving byzantine agreement.

\paragraph*{Partial Synchrony.} Research on Byzantine Agreement in partially synchronous networks has led to significant advancements. Civit et al. \cite{civit2022byzantine} were the first to achieve $\mathcal{O}(n^2)$ communication complexity for BA in this setting. Later, Civit et al. \cite{civit2025partial} proposed the Oper framework, which provides a general method for transforming any synchronous BA protocol into a partially synchronous one while preserving the worst-case per-process bit complexity. However, when applied to the protocol from, e.g., \cite{PODC:CDGGKV24}, Oper only guarantees a total communication complexity of $\mathcal{O}(n^2)$ rather than the adaptive $\mathcal{O}(n \cdot f)$ due to the $\Omega(n)$ per-process bit complexity in the original protocol.

The worst case round complexity of \cite{civit2022byzantine} and \cite{PODC:CDGGKV24} is $\Theta(t)$ and $\Theta(n)$ respectively. 

Beyond standard Byzantine Agreement, State Machine Replication (SMR) is a closely related problem that requires External Validity. One of the prominent SMR protocols in partial synchrony is HotStuff \cite{PODC:YMRGA19}. HotStuff progresses through a sequence of views, each led by a designated leader, and agreement is reached once all honest processes stay in the same view for a sufficient duration. The complexity of achieving view synchronization in HotStuff heavily depends on the clock assumptions of the participating processes. If processes have perfectly synchronized clocks, views can be defined as fixed time intervals, making synchronization trivial. However, in practical settings, clock synchronization may be imperfect. Although HotStuff \cite{PODC:YMRGA19} does not explicitly specify a clock model, it relies on exponential back-off as a mechanism to mitigate clock deviations—an approach that is theoretically sound but practically inefficient. Recent work on view synchronization \cite{lewis-pye2023optimal, lewis2022quadratic, lewis2023fever} introduces a more robust method, ensuring consistent views among processes under partial initial clock synchronization. This approach achieves an $\mathcal{O}(n \cdot f)$ communication complexity and, when integrated with HotStuff, allows the protocol to maintain the same complexity in partial synchrony. Notably, under the assumption of perfectly synchronized clocks, one-shot HotStuff directly achieves $\mathcal{O}(n \cdot f)$ communication complexity without any additional synchronization overhead.

Finally, Spiegelman et al. \cite{spiegelman2020search} established that in partially synchronous networks, any protocol must incur an unbounded number of messages if messages sent before the Global Stabilization Time (GST) are counted. Therefore, when talking about the communication complexity of partially synchronous protocols, we refer only to the messages sent after GST.

\paragraph*{Asynchrony.} The well-known FLP impossibility result \cite{fischer1985impossibility} establishes that Byzantine Agreement cannot be deterministically solved in asynchronous networks if even a single process can fail. However, the use of randomness provides a way to circumvent this limitation. 

For example, Cachin et al. \cite{JC:CacKurSho05} present a protocol that achieves strong BA in asynchrony, ensuring Agreement and Validity unconditionally while guaranteeing Termination with probability 1. This protocol operates with a communication complexity of $\mathcal{O}(n^2)$. 

Further improving on message complexity while allowing a negligible probability of failure, Cohen et al. \cite{cohen2020not} propose a solution for binary BA in asynchrony. Their protocol achieves Agreement, Validity, and Termination with high probability (w.h.p.) and has an expected communication complexity of $\mathcal{O}(n \cdot \log^2 n)$. In parallel, Blum et al. \cite{TCC:BKLL20} were able to achieve BA in asynchrony with overwhelming probability and expected communication complexity $\mathcal{O}(n \cdot \kappa^4)$. Their protocol requires an initial setup phase, which they prove is necessary for subquadratic communication complexity.

All these protocols enjoy constant expected latency. 

Table \ref{tbl:related work} summarizes the existing results and the contribution of the present work.

\begin{table}[ht]
\resizebox{\textwidth}{!}{%
\begin{tabular}{|l|l|l|l|l|l|l|l|}
\hline
\textbf{Paper} & \textbf{Problem}  & \textbf{Communication} & \textbf{Rounds} & \textbf{Resilience} & \textbf{Det.} & \textbf{PKI} & \textbf{Properties} \\
\hline
\multicolumn{8}{|c|}{\textbf{Synchrony}} \\ 
\hline





Civit et al. \cite{PODC:CDGGKV24} & MVBA, IC & $\mathcal{O}(L_o n + n\cdot f)$ & $\mathcal{O}(n)$ & $t < n/2$ & Yes & Yes & A;T;V;IC \\

Lenzen et al. \cite{lenzen2022recursive} & BA & $\mathcal{O}(nt)$ & $\mathcal{O}(f)$ & $t < n/3$ & Yes & No & A;T;V\\

\textbf{This paper} & BA & $\mathcal{O}(nf)$ & $\mathcal{O}(f)$ & $t < n/2$ & Yes & Yes & A;T;V \\

\textbf{This paper} & BA & $\mathcal{O}((n\cdot \log t + tf)\cdot \log n)$ & $\mathcal{O}(f)$ & $t < n/2$ & Yes & Yes & A;T;V \\

Spiegelman \cite{spiegelman2020search} & BA & $\Omega(n + tf)$ & Any & Any & Yes & Yes & A;T;V \\


Abraham et al.\cite{PODC:ACDNPR19} & BA & $\mathbb{E}(\Omega(f^2))$ & Any & Any & No & Yes & A;T;V\\

\hline
\multicolumn{8}{|c|}{\textbf{Partial Synchrony}} \\ 
\hline


Yin et al. \cite{PODC:YMRGA19, lewis2023fever}$^\ast$ & SMR  & $\mathcal{O}(n \cdot f)$ & $\mathcal{O}(f)$ & $t < n/3$ & Yes & Yes & A;T;EV \\

Civit et al. \cite{civit2025partial} & BA &  $\mathcal{O}(n^2)$ & $\mathcal{O}(n)$ & $t < n/3$ & Yes & No & A;T;V \\

\textbf{This paper} & BA & $\mathcal{O}(nf)$ & $\mathcal{O}(f)$ & $t < n/3$ & Yes & Yes & A;T;V \\

\textbf{This paper} & BA & $\mathcal{O}((n + t \cdot f)\cdot \log n \cdot \log t)$ & $\mathcal{O}(f\cdot \log n)$ & $t < n/3$ & Yes & Yes & A;T;V\\

Spiegleman \cite{spiegelman2020search} & BA & $\Omega(\infty)^\S$ & Any & Any & Yes & No & A;T;V \\

\hline
\multicolumn{8}{|c|}{\textbf{Asynchrony}} \\ 
\hline

Cachin et al. \cite{JC:CacKurSho05} & BA & $\mathbb{E}[\mathcal{O}(n^2)]$ & $\mathbb{E}[O(1)]$ &  $t < n/3$ & No & Yes & A;T$^\ddagger$;V\\

Cohen et al. \cite{cohen2020not} & BA & $\mathcal{O}(n \cdot \log^2 n)^\dagger$ & $\mathbb{E}[O(1)]$ & $t < (\frac{1}{3} - \varepsilon)n$ & No & Yes & A$^\dagger$;T$^\dagger$;V$^\dagger$ \\

Blum et al. \cite{TCC:BKLL20} & BA & $\mathbb{E}[\mathcal{O}(n \cdot \kappa^4)]$ & $\mathbb{E}[O(1)]$ & $t < (\frac{1}{3} - \varepsilon)n$ & No & Yes & A$^{\dagger\dagger}$;T$^{\dagger\dagger}$;V$^{\dagger\dagger}$ \\

\textbf{This paper} & BA & $\mathbb{E}[\mathcal{O}((n + t^2)\cdot \log n)]$ & $\mathbb{E}[O(1)]$ & $t < n/3$ & No & Yes & A;T$^\ddagger$;V\\ 

\textbf{This paper} & BA & $\mathbb{E}[\Omega(n + t^2)]$ & Any & Any & No & Yes & A;T$^\ddagger$;V\\

\hline
\end{tabular}
}
\vspace{0.5cm}
\caption{Comparison of Byzantine Agreement Protocols. When a communication is stated with $\Omega$, it indicates a lower bound result. \\
$\ast$ - \cite{lewis2023fever} uses random leader election, but can be derandomized, maintaining its characteristics.\\
A - Agreement, T - Termination, V - Strong Validity, IC - Interactive Consistency, EV - External Validity.\\
$\S$ - unlimited messages before GST, $\dagger$ - with high probability, $\dagger\dagger$ - with overwhelming probability, $\ddagger$ - with probability 1, $\mathbb{E}$ - in expectation.
}
\label{tbl:related work}
\end{table}
\paragraph*{On the use of communication graphs.}
Even in an all-to-all network, it can be useful to restrict communication to a \emph{communication graph} to bound total communication. While promising, this approach remains relatively underexplored for distributed agreement tasks.

Among the works we are aware of, Elsheimy et al.~\cite{elsheimy2024deterministic} use an expander-based communication pattern to obtain adaptive $\mathcal{O}(nf)$ communication for Byzantine Agreement; however, their method inherently has worst-case round complexity $\Omega(n)$. Communication graphs have also been used for leader election, e.g., via bipartite samplers and expanders in~\cite{king2006scalable, bhangale2025leader}. These results, however, guarantee agreement only among most of the honest parties rather than all of them. Notably, the work by Chlebus et al. \cite{chlebus2023deterministic} utilizes expander graphs to disseminate a value known to a subgroup of nodes to the rest of the system. However, their approach is only applicable in synchronous networks where $t =\mathcal{O}(\sqrt{n})$.

\section{Results Overview}

In this section, we formally state our results and outline the main ideas behind them.

Our results are organized along two orthogonal dimensions. First, we consider three standard network models: the synchronous model, the partially synchronous model, and the asynchronous model. Second, we distinguish between two fault regimes: the classical setting where $t = \Theta(n)$, and the large-scale setting where $t \ll n$.

These regimes capture two different bottlenecks. When $t = \Theta(n)$, the central goal is to obtain \emph{adaptive} Byzantine Agreement while still achieving optimal resilience. When $t \ll n$, there is an additional challenge: we aim to \emph{decouple} the dependence on $n$ and $t$, so that systems can scale to many participants without paying a quadratic communication cost in $n$.

\paragraph*{Optimal Resilience Adaptive Algorithms}
For synchronous and partially synchronous settings, we present Byzantine Agreement algorithms with optimal resilience and optimal adaptive communication and round complexities.
\begin{restatable}{theorem}{syncNf}
    \label{thm:sync nf}
    There exists a deterministic algorithm that tolerates up to $t < n/2$ byzantine faults and, given a PKI, solves Byzantine Agreement in synchrony with communication complexity $\mathcal{O}(n \cdot f)$ and round complexity $\mathcal{O}(f)$.
\end{restatable}

\begin{restatable}{theorem}{thmGSTnf}
    \label{theo:gst-main-ba}
    There exists a deterministic algorithm that tolerates up to $t < n/3$ byzantine faults and, given a PKI, solves Byzantine Agreement in partial synchrony with communication complexity $\mathcal{O}(n \cdot f)$ and round complexity $\mathcal{O}(f)$.
\end{restatable}

Interestingly, we show that in asynchrony, adaptiveness is impossible.
\begin{restatable}{theorem}{thmLB}
    \label{thm:lb}
    Let $\mathcal{A}$ be a protocol solving asynchronous byzantine agreement resilient to $t$ byzantine parties with almost sure termination. $\mathcal{A}$ is allowed to have access to $PKI$, to shared randomness, and run any kind of setup, regardless of its message complexity, before parties acquire their input value. Let $M$ be the number of messages exchanged after setup to reach agreement, then there exists an input configuration and a message scheduling protocol such that $\mathbb{E}[M] = \Omega(t^2)$ without any byzantine party. This lower bound holds even if the adversary is static, i.e, it can only corrupt parties at the beginning of the protocol.
\end{restatable}

A notable contribution of this result --- beyond the fact that it holds under a very weak adversary and a strong algorithmic setting --- is that, to the best of our knowledge, it establishes the first lower bound of $\Omega(t^2)$, in contrast to the $\Omega(f^2)$ bounds shown in~\cite{PODC:DolRei82, PODC:ACDNPR19}.

\paragraph*{Large-Scale $t \ll n$ Algorithms.}
In synchrony, partial synchrony, and asynchrony, we provide algorithms whose communication and round complexities have near-optimal dependence on $n$, $t$, and $f$. 
\begin{restatable}{theorem}{mainThmSync}
    \label{thm:sync main}
    There exists a deterministic algorithm that tolerates up to $t < n/2$ Byzantine faults and, given a PKI, solves Byzantine Agreement in synchrony with communication complexity of $\mathcal{O}((n\cdot \log t + t \cdot f)\cdot \log n)$ and round complexity of $O(f)$.
\end{restatable}

\begin{restatable}{theorem}{mainThmGST}
    \label{thm:gst main}
    There exists a deterministic algorithm that tolerates up to $t < n/3$ Byzantine faults and, given a PKI, solves Byzantine Agreement in partial synchrony with communication complexity of $\mathcal{O}((n + t \cdot f)\cdot \log t \cdot \log n)$ and round complexity $O(f\log n)$.
\end{restatable}

\begin{restatable}{theorem}{mainThmAsync}
    \label{thm:async main}
    There exists a randomized algorithm that tolerates up to $t < n/3$ Byzantine faults and, given a PKI, solves Byzantine Agreement in asynchrony with expected communication complexity of $O((n + t^2)\cdot\log n)$ and expected round complexity $O(1)$.
\end{restatable}

Note that although non-adaptive, the communication complexity of $O((n + t^2)\cdot \log n)$ is optimal up to a logarithmic factor due to the infeasibility in Theorem \ref{thm:lb}.

\paragraph*{Techniques}

We build on two simple ideas: (I) broadcasting a value from a small group of nodes, and (II) assigning nodes to committees. While both ideas have been used extensively in the theoretical and practical literature, we substantially strengthen what is known: we provide a \emph{deterministic} scheme that is therefore secure even against an \emph{adaptive adversary}. To the best of our knowledge, prior approaches either rely on randomness and/or assume weaker adversaries that cannot corrupt parties on the fly.

In particular, we leverage \emph{dispersers}, defined in Section \ref{sec:disperser}, to obtain robust committee assignment (note that parties can be in multiple committees):

\begin{restatable}{definition}{defCommittee}
    We call a committee \emph{compromised} if it contains at least one byzantine party. 
    
    For a given assignment of parties to committees and a given set of byzantine parties, we say that a party is \emph{blocked} if all committees it is assigned to are compromised. 
\end{restatable}

\begin{restatable}{theorem}{thmNodesToCommittees}
    \label{thm:nodes2committees}
    Given $n$ parties and an integer $\hat{f}$, there is a way to assign parties to committees such that:
    \begin{enumerate}
        \item There are at most $O(\hat{f} \log n)$ committees.
        \item Every party is assigned to at most $O(\log n)$ committees.
        \item An adversary, after seeing the assignment, can make up to $\hat{f}$ arbitrary parties byzantine. Nonetheless, no matter the adversarial choice of which parties to corrupt, the number of blocked parties is at most $c_b\cdot \hat{f}$ for some constant $c_b$.
    \end{enumerate}
\end{restatable}

We would like to remark that \Cref{thm:nodes2committees} leverages a disperser graph whose existence was only proven non-constructively with probabilistic methods. Using explicit dispersers that can be constructed deterministically in polynomial time results in a bigger polylog factor in the communication complexity of our algorithms.

\section{Preliminaries}
We consider a set of $n$ parties $\mathcal{P} = \{ p_1, \dots, p_n \}$ running a distributed protocol in a fully connected network. Channels are authenticated, meaning that when receiving a message, a party knows who sent it. 

\noindent\textbf{Network Setting}: We consider $3$ different network settings:
\begin{itemize}
    \item \textit{Synchronous network}: We assume there is a known $\Delta > 0$ such that when sending a message, a party has the guarantee that this message will be received within $\Delta$ units of time. This allows protocols to be expressed in rounds where each round takes $\Delta$ units of time.
    \item \textit{Asynchronous network}: There is no guarantee on the delay between a party sending a message and it being received. The only guarantee is that a message sent is eventually received.
    \item \textit{Partially synchronous network}: We assume the setting of \cite{PODC:DwoLynSto84}: There is a known $\Delta > 0$ and an unknown \emph{Global Stabilization Time} (GST), such that a messages sent at time $t$ are received by time at most $\max\{t, GST\} + \Delta$.
\end{itemize}

\noindent\textbf{Byzantine Behavior}: We assume that a set of at most $t$ parties can exhibit arbitrary malicious behavior when running the protocol. In particular, an adversary is able to control these parties, and the protocol must be resilient to this setting. Parties that are not corrupted are called honest. In the rest of this paper, we will also consider $f$, the number of actual byzantine parties during an execution of a protocol: $0 \leq f \leq t \leq n$.

Moreover, we assume that for the asynchronous setting as well as the partially synchronous setting before GST, the adversary has full control over the scheduler and can delay or reorder any messages as long as they are eventually received.

Our protocols are resilient against a dynamic adversary. This means that the adversary is allowed to observe messages being sent before deciding which parties to corrupt.

\noindent\textbf{Clocks and Views}: For our partially synchronous protocols, we adopt a ubiquitous concept of a \emph{view}. A view refers to a logical phase of the protocol during which a specific process is designated as the leader responsible for driving progress. Each view is associated with a unique view number, and all correct nodes share a common function $\mathrm{leader}(i)$ that defines a leader of the $i$-th view. If the leader fails to make progress — due to crashes, equivocation, or network delays — the protocol initiates a view change, advancing to the next view with a new leader. With this approach, different nodes might be in different views, ignoring leaders other than those of their view, hence preventing them from obtaining enough votes/signatures/etc. to drive the progress. 

The view abstraction must satisfy the following property: \textit{eventually after GST, all honest nodes stay in the view with an honest leader for long enough}. Here, the ``long enough'' part is protocol-specific, and in our partially synchronous protocols, this always refers to a time period of $\Theta(\Delta)$.

Note that implementing a view abstraction is a straightforward process if nodes have perfectly synchronized clocks. In this case, one can define an $i$-th view to be a time period of, for instance, $[4i\cdot\Delta, 4(i + 1)\cdot\Delta]$, hence making each view span a time period of $4\Delta$. An alternative to a perfect clocks assumption is a \emph{Partial Initial Clock Synchronization} assumption \cite{lewis2023fever}, which states that there is a known time period $\Gamma$ such that at least $t$ honest processes start a protocol within time $\Gamma$ after the first honest process starts it. Under this assumption, Fever \cite{lewis2023fever} implements a view change in a way that at most $\mathcal{O}(n + n \cdot f)$ words are sent before all honest nodes stay in the same view with an honest leader for long enough. We remark that Fever can be made to send only $\mathcal{O}(n + t \cdot f)$ words by making only $2t + 1$ instead of $n$ processes sending a $view$ message to the leader. Another common clock assumption is to say that local clocks of the processes can drift arbitrarily before GST, but do not drift after \cite{civit2025partial}.

In our partially synchronous algorithms, we abstract out both the clocks and the internals of a view change.

\noindent\textbf{Byzantine Agreement}: In this paper, we focus on solving the problem of Strong Binary Byzantine Agreement, which we call just Byzantine Agreement for simplicity. In this problem, each honest party can \emph{propose} a value (either $0$ or $1$) and can \emph{decide} a value. The Byzantine Agreement is defined by the following properties.
\begin{itemize}
    \item \textit{Termination}: Every honest party eventually decides a value.
    \item \textit{Agreement}: If two honest parties $p_1$ and $p_2$ respectively decide values $v_1$ and $v_2$, then $v_1 = v_2$.
    \item \textit{Strong Unanimity}:\footnote{In the literature, this property is also called Strong Validity.} If all honest parties propose the same value, then this value must be decided.
\end{itemize}
For the synchronous and partially synchronous setting, we say that a protocol satisfies Byzantine Agreement (BA) if it satisfies Termination, Agreement, and Strong Unanimity. Because of the FLP impossibility result \cite{fischer1985impossibility} stating that Byzantine Agreement with aforementioned properties can not be solved in asynchrony, for the asynchronous setting, we replace the Termination with a \emph{Probabilistic Termination} property:
\begin{itemize}
    \item \textit{Probabilistic termination}: Every honest party eventually decides a value almost surely.
\end{itemize}

\noindent\textbf{Cryptographic Primitive}: In some of our protocols, we will assume the presence of a public-key infrastructure (PKI). For a party $p$, the interface of PKI consists of functions $\textit{sign}_p$ and $\textit{verify}$ such that:\footnote{As is standard in cryptography, here and in similar contexts, we assume only party $p$ can invoke $\textit{sign}_p(\cdot)$, but all parties can invoke $\textit{verify}(\cdot, \cdot, p)$.}
\begin{itemize}
    \item For a given party $p$ and message $m$, $\textit{sign}_p(m)$ returns a signature $s$.
    \item For a given party $p$, value $s$ and message $m$, $\textit{verify}(m,s,p)$ returns $true$ if and only if $s$ was obtained from calling $\textit{sign}_p(m)$.
\end{itemize}
Moreover, for some protocols, we will also assume the presence of a threshold signature scheme \cite{JC:BonLynSha04}. A threshold signature scheme is parameterized by two quantities: $k$ - the threshold, and $n$ - the number of parties participating in the scheme. Each party $p$ that participates in the scheme is given functions $\textit{tsign}_p$, $\textit{tcombine}$, and $\textit{tverify}$ such that:
\begin{itemize}
    \item For a given party $p$ and message $m$, $\textit{tsign}_p(m)$ returns a partial signature $\rho_p$.
    
    \item For given message $m$, a set $P \subseteq \mathcal{P}$ such that $|P| = k$ and tuple $(\rho_p)_{p \in P}$ such that $\forall p \in P,\ \rho_p = \textit{tsign}_p(m)$, a call $\textit{tcombine}\left(m, (\rho_p)_{p \in P}\right)$ returns a threshold signature $\sigma$.
    
    \item For a given message $m$ and threshold signature $\sigma$, $\textit{tverify}(m, \sigma)$ returns $true$ if and only if there exists a tuple of partial signatures $(\rho_p)_{p \in P}$ with $|P| = k$ such that $\sigma = \textit{tcombine}\left(m, (\rho_p)_{p \in P}\right)$.
\end{itemize}
Except when mentioned otherwise, the threshold signature will be used with parameter $k = n - t$ in this paper.

We will also use \emph{aggregate signature scheme} \cite{EC:BGLS03}, which, for simplicity of presentation, we define to be a threshold signature scheme\footnote{The original definition is more general; it allows parties to sign different messages. We do not need this versatility in our results.} with $k = n.$ 

\noindent\textbf{Communication Complexity.} We define the \emph{communication complexity} of an algorithm as the number of \emph{words} sent by honest nodes, where each word contains a constant number of signatures and a logarithmic number of bits. In our protocols, all messages contain a constant number of words, so we use word complexity and message complexity interchangeably. 

As pointed out by Spiegelman et al. \cite{spiegelman2020search}, there does not exist a partially synchronous algorithm that solves BA sending a bounded number of messages when counting messages sent before GST. Therefore, when speaking about the communication complexity of partially synchronous algorithms, we refer only to messages sent \emph{after} GST. 

\noindent\textbf{Round Complexity.} We define the round complexity of a synchronous Byzantine Agreement algorithm to be the first round by which all honest nodes have decided a value. The round complexity of asynchronous algorithms is defined assuming they run in synchrony from the start. Finally, for partially synchronous algorithms, the round complexity is the first round by which all honest nodes have decided a value minus the round when GST is reached.

\section{High-Level Structure of the Algorithms}
All our algorithms that decouple $n$ and $t$ in the communication complexity have the same high-level structure, which we now present in this section. It goes as follows. First, we pick a \emph{Quorum} of $O(t)$ nodes and run agreement among those using some efficient agreement algorithm. We call this phase of the protocol \emph{quorum agreement}. The second phase is disseminating the decided value from quorum nodes to the rest of the system. This phase we call \emph{Quorum to All Broadcast} or QAB for short.

\subsection{Quorum Agreement}
\label{sec:qba}
The first step in our protocol is to reach agreement among a quorum. The specific agreement mechanism depends on the network model. In the asynchronous setting, we use the protocol of Cachin et al. \cite{JC:CacKurSho05}. For synchronous and partially synchronous networks, however, prior work did not provide algorithms that simultaneously achieve our targets for adaptive communication and round complexity. We therefore present new algorithms meeting these requirements, which are fully described in \cref{sec:leader-based-agreement}.

We use a view-based approach based on Spiegelman's byzantine agreement protocol for external validity \cite{spiegelman2020search}, but with many abstractions and modifications to make it fit for strong agreement in both the synchronous and partially synchronous setting with optimal efficiency:
\begin{itemize}
    \item \textbf{Proofs and threshold signature}: The original protocol relies on creating multiple proofs based on $(n-t)$-threshold signatures, each stronger than the precedent, until getting a commit proof, which satisfies agreement. We change the threshold from $n-t$ to $\lceil (n + t + 1)/2\rceil$, similar to what is used in \cite{cohen2023make}. This has no impact in the partially synchronous setting but guarantees that we get the same intersection properties for the synchronous setting. A consequence of this is that if $f \geq \lfloor (n - t - 1)/2\rfloor$, which can happen in the synchronous setting, then it is not guaranteed that a value will be decided, even if the leader is honest.
    \item \textbf{Abstracting the valid input}: In external validity, we assume each party to hold an input that (1) is valid to agree on and (2) each party can locally check that this value can be agreed on. We abstract this further to work with strong unanimity by assuming that parties no longer hold such input at the beginning of the protocol, but instead have at their disposal a \textbf{retrieval} sub-protocol which can retrieve such input at the cost of some additional communication complexity and rounds. We also assume that this sub-protocol can fail and return $\bot$. This way, the only difference between the synchronous and partially synchronous settings on this part of the algorithm is the retrieval protocol implementation.
    \item \textbf{Always making progress in synchrony}: In the synchronous settings, because of the high resiliency requirement $t < n/2$, it might happen that no proof can be generated if the honest parties' input is evenly distributed between $0$ and $1$. To handle this setting, we make small progress every time the leader is honest with a simple observation: If the leader cannot make a proof, then the leader knows that both $0$ and $1$ can be decided.
    \item \textbf{Complaining round}: We note that the original paper's round complexity is $\mathcal{O}(n)$. To reduce the round complexity to $\mathcal{O}(f)$, we add a complaining round inside each view. This idea is similar to the catch-up mechanism from the DARE to agree protocol \cite{PODC:CDGGKV24}. In essence, as long as a party has not decided a value, it will complain to the leader at every view. This ensures that every honest party will have decided after a view where the leader is honest while only having a $\mathcal{O}(n\cdot f)$ overhead.
    \item \textbf{Synchronous fallback}: As stated above, if $f < \lfloor (n - t - 1)/2\rfloor$ in the synchronous setting, our previous approach is guaranteed to terminate. Otherwise, we remark that $f = \Omega(n)$ and we use the same technique as \cite{cohen2023make}, which is to detect this situation and run a fallback agreement protocol with quadratic message complexity. We make improvements to this approach to ensure no honest party will send any messages after $\mathcal{O}(n)$ rounds and by simplifying the fallback call.
\end{itemize}

\subsection{Quorum to All Broadcast}
\label{sec:qab}
The second step is to efficiently broadcast the value agreed upon from quorum nodes to the rest of the system: to perform a QAB. Protocols and proofs are provided in \cref{sec:quorum-to-all-sec}.

The formal definition of the primitive goes as follows. A QAB instance is defined with respect to a group of dedicated nodes called \emph{quorum}. Every honest node $u$ in the quorum starts an instance with a value $\vin$ (common among all quorum nodes) and a proof $proof_u$. We assume that all the honest nodes in the system possess a predicate $\mathrm{verify}(\cdot, \cdot)$ such that $\mathrm{verify}(v, proof)$ is true only if $v = \vin$ and for any honest quorum node $u$, $\mathrm{verify}(\vin, proof_u)$ is true. QAB exposes the following interface:

\begin{itemize}
    \item $\textsc{QABInitiate}()$ - start the protocol for any party.
    \item $\textsc{StartQuorum}(\vin, proof)$ - provide a value and proof to a quorum node. 
    \item $\mathrm{Decide}(v)$ - decide value $v$.
\end{itemize}

QAB is defined by the following property that encapsulates classical Agreement, Termination, and Validity. 
\begin{property}[Complete Correctness]
    Every honest node eventually decides $v_{in}$.
\end{property}

In our applications, it is crucial that quorum nodes might not possess $\vin$ at the beginning of the execution, but acquire it at some arbitrary point in time. Hence, our QABs support quorum nodes invoking $\textsc{StartQuorum}(\vin, proof)$ at different times. This motivates us to define the \emph{round complexity of QAB} as the number of rounds between the last honest quorum node invokes $\textsc{StartQuorum}(\vin, proof)$ and all honest nodes decide the value.

As for non-quorum nodes, there is no specific reason for them to invoke $\textsc{InitiateQAB}()$ later; hence, they invoke it at the very beginning. This brings another challenge: we need to make sure that non-quorum nodes do not exchange a lot of messages before the problem can even be solved, that is, before any quorum node invokes $\textsc{StartQuorum}(\vin, proof)$. 

\paragraph*{QAB Phase} The main building block common among our QABs in different network models is called a \emph{QAB Phase}. It is a procedure parameterized by a value $\hat{f}$. On a high level, a QAB phase with parameter $\hat{f}$ is guaranteed to succeed in disseminating a quorum value to all parties if the actual number of faults $f$ is no larger than $\hat{f}$. We now give a description of a QAB phase. 

We restrict the communication between nodes to a specific graph, which we call a \emph{communication graph}. If there is an edge between two nodes in this graph, we say that these two nodes are \emph{linked}. 

We construct the communication graph $G$ as follows. First, we assign parties to committees according to Theorem \ref{thm:nodes2committees} with parameters $n$ and $\hat{f}$. This way, we obtain a set $\mathcal{C}$ of $O(\hat{f}\log n)$ committees with each party being assigned to $O(\log n)$ committees. For each committee $C \in \mathcal{C}$, we choose (arbitrarily) a node in $C$ ``responsible'' for this committee and call such a node a \emph{relayer} for $C$, denoted $relayer(C)$. Hence, we have $O(\hat{f}\log n)$ relayers. In the communication graph, we then link every party to the relayer of each committee this party is assigned to. That is $\{(u, relayer(C))\mid u \in C, C \in \mathcal{C}\}  \subseteq E(G)$. We also link every relayer to every quorum node.

\begin{figure}
    \centering
    \includegraphics[width=0.5\linewidth]{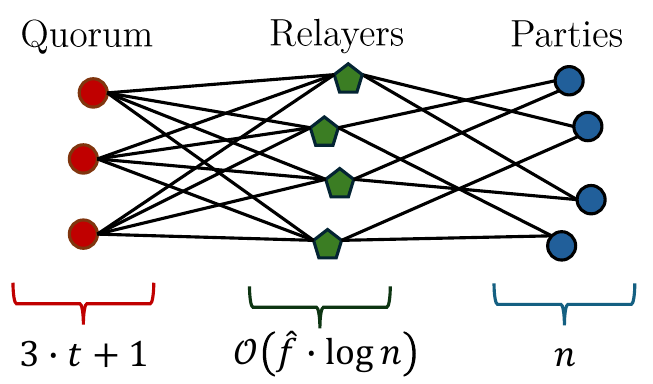}
    \caption{The communication graph for communications with relayers. Nodes are assigned three roles (non-exclusive): \textit{parties}, \textit{relayers} and \textit{quorum} with respective sizes of $n$, $O(\hat{f}\log n)$ and $3t + 1$. Each party is only linked to $O(\log n)$ relayers, and each relayer is linked to each quorum node.}
    \label{fig:communication graph}
\end{figure}

Given the communication graph, the high-level description of our QAB phase with parameter $\hat{f}$ is as follows. The pseudocode for the procedure can be found in \cref{sec:quorum-to-all-sec}.

\begin{enumerate}  
    \item \textbf{Quorum-Relayers-Parties.} Each quorum node, upon invoking $StartQuorum(\vin,proof)$, broadcasts $(\vin, proof)$ to every relayer. Upon receiving and verifying $(\vin, proof)$, each relayer then broadcasts it among the parties linked to them. 
    \item \textbf{Parties-Relayers-Quorum.} Every committee has an aggregate signature setup. After receiving $(\vin, proof)$ from a relayer and verifying it, each party produces a partial signature for a message saying that it knows $\vin$ and sends it back to the relayer. After receiving a valid partial signature from every party linked to them for a given committee, a relayer aggregates the signatures into a single one and sends it to all the quorum nodes from whom it received $\vin$.
    \item \textbf{Quorum-Undecided.} Each quorum node $u$ maintains a set $Acknowledged_u$, initially empty, of nodes of whom $u$ is aware that they know $\vin$. Once $u$ receives and verifies an aggregated signature for a committee $C$ from a $relayer(C)$, $u$ knows that all the parties in $C$ know $\vin$ and hence it updates $Acknowledged_u = Acknowledged_u \cup C$. Once $|Acknowledged_u| \geq |n - c_b\cdot \hat{f}|$ where $c_b$ is a constant from Theorem \ref{thm:nodes2committees}, $u$ sends $\vin$ directly to the nodes in $\mathcal{P} \setminus Acknowledged_u$ and terminates.
\end{enumerate}

As a remark, note that in the last step, when $u$ is sending $\vin$ to the parties in $\mathcal{P} \setminus Acknowledged_u$, this happens \emph{not} according to the communication graph. For $\mathcal{P} \setminus Acknowledged_u$ is determined dynamically. This is the only such place; all other messages are sent according to the graph. 

\subsection{Putting it all together}

The final step is to show that using a regular byzantine agreement along with a QAB protocol, one can obtain an improved protocol for byzantine agreement which decouples $n$ and $t$:
\begin{theorem}\label{theo:decoupling}
    Given a network setting, we assume there is a $BA$ protocol for binary agreement with message complexity $\textsc{Messages}_{BA}(n, f)$ supporting with resiliency $t < n/3$ and round complexity $\textsc{Rounds}_{BA}(f)$ which outputs a proof that a value was decided along with the value. We assume as well there exists $QAB$ protocol with message complexity $\textsc{Messages}_{QAB}(n,t,f,l)$, where $l$ is an upper bound on the number of rounds before every honest quorum node calls $\textsc{StartQuorum}(\vin, proof)$, and round complexity $\textsc{Rounds}_{QAB}(n, f)$. 

    Then there exists a $BA$ protocol with resiliency $t < n/3$, message complexity $\textsc{Messages}_{BA}(3t+1,f) + \textsc{Messages}_{QAB}(n,t,f,\textsc{Rounds}_{BA}(f))$ and round complexity $\textsc{Rounds}_{BA}(f) + \textsc{Rounds}_{QAB}(f)$.
\end{theorem}

\begin{proof}

Our protocol \textsc{DecoupledBA}, whose pseudocode is available below relies on running the BA protocol $\Pi_{BA,Q}$ on a set $Q$ of $3t+1$ parties which make up the quorum nodes, then to use the QAB protocol to spread the output of this BA along with its proof to every party.

\begin{dianabox}{\textsc{DecoupledBA}}
\algoHead{$BA$ protocol with decoupled $n$ and $t$ for a party $p$ with input $\vin$}
\begin{algorithmic}[1]
    \InParallel{}
        \State Call $\textsc{QABInitiate}()$
        \State When QAB decides a value $v$, decide $v$
    \EndParallel
    \State Let $Q$ be the set of the $3t+1$ parties with lowest IDs
    \If{$p \in Q$}
        \State $(value, proof) \gets \Pi_{BA,Q}(\vin)$
        \State Call $\textsc{StartQuorum}(value, proof)$
    \EndIf
\end{algorithmic}
\end{dianabox}

The round complexity of \textsc{DecoupledBA} is $\textsc{Rounds}_{BA}(f) + \textsc{Rounds}_{QAB}(f)$, moreover, all honest quorum nodes will call $\textsc{StartQuorum}(\vin, proof)$ at most $\textsc{Rounds}_{BA}(f)$ rounds after the start of the protocol. Therefore, the message complexity is $\textsc{Messages}_{BA}(3t+1,f) + \textsc{Messages}_{QAB}(n,t,f,\textsc{Rounds}_{BA}(f))$. Because of the BA termination and agreement, all honest quorum nodes will call $\textsc{StartQuorum}(\vin, proof)$ with the same $\vin$ and a proof for it, where $\vin$ satisfies strong unanimity because of the BA validity. Because of QAB complete correctness, all parties will therefore all eventually decide the same value $\vin$. So termination, agreement and validity are satisfied.

\end{proof}

\section{Disperser}
\label{sec:disperser}
The core building block in all our algorithms is a \emph{disperser} graph. 
\begin{definition}
    Consider a bipartite graph $G = (L \sqcup R, E)$ where each node on the left has degree $d$. $G$ is a $(k, \varepsilon)$ \emph{disperser} if for every $S \subseteq L$ of size at least $k$, it holds that $|N(S)| \geq (1 - \varepsilon)|R|$.

    Here $N(S)$ is a neighborhood of $S$, namely $N(S) = \{v \in R\mid (u, v) \in E\text{ for some } u\in S\}$.
\end{definition}

Informally, a bipartite graph is a disperser if any large enough subset on the left almost covers the whole right part. Yet, this property alone can be trivially achieved by a complete bipartite graph. Also, it is not interesting if the subset on the left covers the right part just because it is huge itself. Therefore, what's challenging is to achieve the defining disperser property with few edges (ideally with $d = O(\log n)$) and for subsets small enough compared to the right side.

Fortunately, graphs satisfying these conditions exist. The following is a known fact (see, e.g., Theorem 1.10 in \cite{radhakrishnan2000bounds}).

\begin{theorem}
    \label{thm:bipartite disperser}
    For any $n, m$, $m \leq n$, there exists a $(k, 1/2)$ bipartite disperser $G(L\sqcup R, E)$ with $|L| = n, |R| = m$, left degree $d = O(\log n)$ and $k = O(\frac{m}{\log n})$. 
\end{theorem}

However, the existence of bipartite dispersers in Theorem \ref{thm:bipartite disperser} has only been demonstrated non-constructively to date. Known explicit bipartite dispersers, e.g., those given by Ta-Shma et al. \cite{ta2007lossless}, result in polylogarithmically worse guarantees:
\begin{theorem}[Theorem 1.4~\cite{ta2007lossless}]
    \label{thm:explicit bipartite disperser}
    There exists a deterministic polynomial time algorithm that for any $n, m$, $m \leq n$ outputs a $(k, 1/2)$ bipartite disperser $G(L\sqcup R, E)$ with $|L| = n, |R| = m$, left degree $d = O(\mathrm{poly}(\log n))$ and $k = O(\frac{m\log ^3 n}{d})$. 
\end{theorem}

For the rest of the paper, we will use bipartite dispersers stated by Theorem \ref{thm:bipartite disperser}. Switching to the explicit version of Theorem \ref{thm:explicit bipartite disperser} would result in additional multiplicative polylogs in the communication complexity of our algorithms.  

Let us now discuss a central theorem that paves the way for using bipartite dispersers in distributed computing. At its core, the theorem shows how to assign parties to committees so that, even if the adversary chooses which parties to corrupt after seeing the assignment, most committees contain no byzantine parties.

We will need the following definitions. 
\defCommittee*

Now, the theorem can be stated as follows:
\thmNodesToCommittees*
\begin{proof}
    We construct the stated assignment according to the bipartite disperser graph, interpreting left nodes of the graph as parties, right nodes as committees, and an edge as a party being assigned to a committee. 

    To give a precise argument, we need to unpack constants hidden by O-notation in Theorem \ref{thm:bipartite disperser}. In particular, Theorem \ref{thm:bipartite disperser} states that there are constants $c_d$ and $c_k$ such that for all integers $N, M$ with $M \leq N$ there exists a $(k, \frac{1}{2})$ disperser with $N$ nodes on the left, $M$ nodes on the right, left degree $d \leq c_d\cdot\log N$ and $k \leq c_k\cdot\frac{M}{\log N}$.

    Now, consider such a disperser for $N = n$ and $M = 4c_d\cdot \hat{f}\cdot\log n$.

    Applying the definition of a bipartite disperser, we get:

    \begin{enumerate}[label=Observation \Roman*., leftmargin=*, itemsep=0.6\baselineskip]
    \item Any subset of parties of size at least
      $k = 4c_kc_d\cdot \hat{f}$
      is assigned to at least
      $M/2 = 2c_d\cdot \hat{f}\cdot \log n$
      committees.
    
    \item Furthermore, since one Byzantine party compromises at most $d$ committees,
      and there are at most $\hat{f}$ Byzantine parties, the number of compromised
      committees is at most
      $f\cdot d = c_d\cdot \hat{f} \cdot \log n$.
    \end{enumerate}
    
    As a final step, assume for contradiction that there exists a set $S$ of parties
    of size at least $4c_kc_d\cdot f$ in which every party is blocked. However, by
    Observation~I, parties in $S$ are assigned to at least $2c_d\cdot f\cdot \log n$
    committees, and by Observation~II, at most $c_d\cdot f \cdot \log n$ committees
    are compromised; therefore, some party in $S$ is assigned to a non-compromised
    committee, contradicting that all parties in $S$ are blocked. Letting
    $c_b := 4c_kc_d$ completes the proof.
\end{proof}

\section{Leader-Based Agreement}\label{sec:leader-based-agreement}

In this section, we present a protocol relying on a leader-based approach to achieve an adaptive communication complexity of $\mathcal{O}(n\cdot f)$. This protocol assumes the existence of what we call a retrieval protocol, defined below, which allows a leader to get a proof that a value can be decided. Our protocol only assumes a partially synchronous setting. In fact, by using different retrieval procedure, we can get optimal resiliency for strong unanimity in both the synchronous and partially synchronous setting. Our algorithm is based on a modified version of Spiegelman's byzantine agreement protocol for external validity \cite{spiegelman2020search}, which itself uses Hotstuff's view synchronization \cite{PODC:YMRGA19} technique. While inspired by Spiegelman's protocol for network-agnostic settings \cite{spiegelman2020search}, our work departs significantly in both design and guarantees. We eliminate the asynchronous and randomized components of the original protocol and instead strengthen the synchronous part with additional checks, enabling it to function effectively in partial synchrony after GST. This modification allows us to achieve $\mathcal{O}(n \cdot f)$ message complexity for external validity without relying on full synchrony.

Crucially, by allowing parties to selectively share their commit proofs upon request, we ensure termination in $\mathcal{O}(f)$ rounds while keeping the message complexity at $\mathcal{O}(n \cdot f)$.

We provide an overview of each step and property of the algorithm. For detailed pseudocode, please refer to Appendix \ref{sec:apx:alg gst pseudocode}.

The protocol is organized in views. Each view has a leader, who is chosen in a round-robin way. As a consequence of Lemma \ref{lemma:gst-honest-commit}, if $f \leq \lfloor \frac{n - t -1}{2} \rfloor$, the algorithm is guaranteed to succeed after at most $2f+1$ views after GST. 

\begin{definition}
    We consider two subprotocols \textsc{RetrieveLeader} and \textsc{RetrieveParty}($\vin$). \textsc{RetrieveLeader} is called by a leader while \textsc{RetrieveParty} is called by every honest party. We assume parties know who the leader is and if the leader is honest, then all honest parties call \textsc{RetrieveParty} at most one round after the leader called \textsc{RetrieveLeader}. A call to \textsc{RetrieveLeader} along with its associated \textsc{RetrieveParty} is called an instance. We say that $(\textsc{RetrieveLeader},\textsc{RetrieveParty}(\vin))$ is a retrieval procedure if it satisfies the following properties
    \begin{itemize}
        \item \textsc{RetrieveLeader} either returns $\bot$ or a tuple $(v, proof)$.
        \item In the case the leader returns a tuple $(v, proof)$. $v$ should be a value ($0$ or $1$) that satisfies strong unanimity with $\vin$ and $proof$ can be used by any honest party to locally verify that $v$ can be decided (similar to external validity). Moreover, the output must fit within a constant number of words, i.e $proof$ should only contain a constant amount of signatures and/or threshold signatures.
        \item An instance has a $\mathcal{O}(n)$ message complexity.
        \item When called after GST, the protocol returns within $2$ rounds if the leader is honest.
        \item If called after GST by $f+1$ different honest leaders, this protocol will at least once \textbf{not} return $\bot$.
    \end{itemize}
\end{definition}

Let $k = \lceil \frac{n+t+1}{2}\rceil$. The choice of $k$ is made so that any two sets of $k$ parties have at least one honest party in common. After the model-specific retrieval phase, each party keeps track of the following values:
\begin{itemize}
    \item Key (initially $\bot$): the key is expressed as a value, a view number, and a proof ($(n, k)$-threshold signature). It indicates a value that is safe to choose. 
    \item Lock (initially $\bot$): similar to the key, the lock is expressed as a value, a view number, and a proof. It is used to check the key value proposed by a leader.
    \item Commit (initially $\bot$): the commit is expressed as a value and a proof. We will prove that at most one commit value can be created when running the protocol, and that if a commit value is created, this value will eventually be decided by all honest parties.
\end{itemize}

The algorithm for a single view consists of $4$ phases of a leader asking and other parties responding. A leader proceeds to the next phase once it get the appropriate $k$ partial signatures to generate the proof for the next step. After $4$ phases pass successfully, the leader obtains a commit value and broadcasts it.

We give a brief overview of the role of each phase in a view:
\begin{itemize}
    \item \textbf{Suggestion phase}: The leader requests a value, and the parties respond with a commit value, then a key with the highest view number if they have no commit value, or nothing if they have no commit and no key value. Parties also run the retrieval protocol for the leader to get a valid value to agree on or $\bot$.

    The purpose of this phase is to enable the leader to determine whether any party has already decided (via a commit proof) or was close to deciding (via a key proof). If neither is the case, the leader falls back to the value given by the retrieval protocol. If the retrieval protocol also returns $\bot$, the leader stays silent for the rest of the view.
    
    \item \textbf{Key phase}: 
    The leader proposes either a key with the highest view it observed, or if no key was observed, the value returned by the retrieval protocol, which we assume is not $\bot$ as the leader would otherwise be silent. A party accepts the proposal if it has no lock, or if the proposal includes a key and their current lock has a lower view number than the proposed key. The leader then aggregates the positive responses to generate a key proof.

    The purpose of this phase is to ensure that the value proposed by the leader is consistent with the information parties obtained in earlier rounds. Note that parties compare the proposed key to the lock they have, not to another key; this is to ensure agreement after GST.
    
    \item \textbf{Lock phase}: The leader broadcasts a key proof, parties verify and respond with their accepts, and by aggregating those, the leader creates a lock proof. 
    
    The goal of this phase is for parties to learn about the key proof, which may be used for the following views to ensure agreement.
    
    \item \textbf{Commit phase}: The leader broadcasts a lock proof, parties verify and respond with their accepts, and by aggregating those, the leader creates a commit proof.
    
    The goal of this phase is for parties to learn the lock proof, which is used to check the value proposed by the leader in the key phase.
    
    \item \textbf{Commit broadcast:} The leader broadcasts the commit proof to every party, who can decide it upon receiving.
\end{itemize}

\begin{figure}[h]
    \centering
    \includegraphics[width=\textwidth]{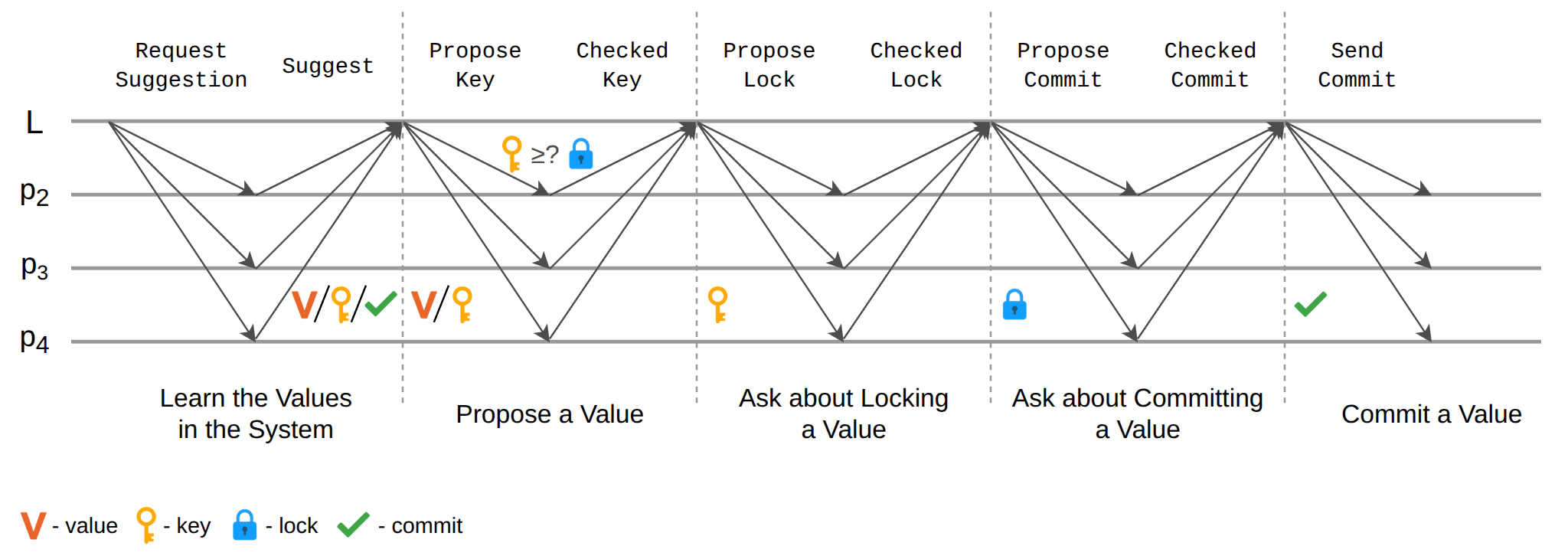}
    \caption{Evolution of a single view in our partially synchronous Byzantine Agreement protocol, assuming an honest leader L and that the system is after GST.}
    \label{fig:gst_alg}
\end{figure}

This algorithm, to which we refer as \textsc{ViewByzantineAgreement} and which is described in \cref{sec:apx:alg gst pseudocode}, solves Byzantine Agreement assuming the actual number of byzantine $f$ is low enough:

\begin{restatable}{theorem}{ViewBaMain}\label{thm:view-ba-main}
    Assuming $f \leq n-k$ (i.e $f \leq \lfloor \frac{n-t-1}{2} \rfloor$) and the existence of a retrieval protocol, then \textsc{ViewByzantineAgreement} achieves agreement, termination and validity. It also returns a proof which can only be obtained for the decided value and has message complexity $\mathcal{O}(nf)$ and round complexity $\mathcal{O}(f)$.
\end{restatable}

\section{Quorum-to-All Phase}
\label{sec:quorum-to-all-sec}

We now give proper definition of a Quorum-to-All (QAB) phase which is used for all $n-f$ honest node to learn the decided value from the quorum of size $3t+1$. A QAB phase is characterized by a parameter $\hat{f}$. The main idea is that if the phase is ran and $f \leq \hat{f}$, then by the end of the phase, all honest nodes will have decided while using only $\mathcal{O}(n \log n + t \cdot f \cdot \log n)$ words and spending a constant amount of rounds.

As said previously, nodes can assume multiple roles including the ones of party, relayer and quorum node. In particular, a single node may be the relayer for multiple committees. We now give an implementation of the $\textsc{QABInitiate}()$ procedure as well as the pseudocode used for parties (\textsc{QABParty}), relayers (\textsc{QABRelayer}) and quorum nodes (\textsc{QABQuorum}).

\begin{dianabox}{\textsc{QABInitiate}}
\algoHead{Protocol for QAB initialization for a party $p$}
\begin{algorithmic}[1]
    \State Run \textsc{QABParty} in parallel
    \For{every committee $C$ where $p$ is the relayer}
        \State Run $\textsc{QABRelayer}(C)$ in parallel
    \EndFor
\end{algorithmic}
\end{dianabox}

\begin{dianabox}{\textsc{QABParty}}
\algoHead{Protocol for a party $p$ running in a QAB phase}
\begin{algorithmic}[1]
    \For{Each committee $C$ which $p$ is a part of, with relayer $r$}
        \UponTrue{Receiving valid $(AGGREGATE, v, proof)$ from $r$}
            \State Send $(ACK, sign_C(ACK))$ to $r$
            \If{$p$ has not yet decided}
                \State Decide $v$
            \EndIf
        \EndUpon
    \EndFor
    \Statex

    \UponTrue{Receiving valid $(DECIDE, v, proof)$}
        \If{$p$ has not yet decided}
            \State Decide $v$
        \EndIf
    \EndUpon
\end{algorithmic}
\end{dianabox}

\begin{dianabox}{\textsc{QABRelayer}}
\algoHead{Protocol for a relayer $r$ of a committee $C$}
\begin{algorithmic}[1]
    \State $Aggregate \gets \bot$
    \State $Signatures \gets \emptyset$
    \State $Waiting \gets \emptyset$
    \UponTrue{Receiving valid $(DISPERSE, v, proof)$}
        \For{Party $p \in C$}
            \State Send $(AGGREGATE, v, proof)$ to $p$
        \EndFor
    \EndUpon
    \Statex
    \UponSimple{Receiving valid $(DISPERSE, v, proof)$ from a quorum node $q$ for the first time}
        \If{$Aggregate \ne \bot$}
            \State Send $(COMPLETED, Aggregate)$ to $q$
        \Else{}
            \State $Waiting \gets Waiting \cup \{q\}$
        \EndIf
    \EndUpon
    \Statex

    \UponSimple{Receiving valid $(ACK, sign)$ from a committee party $p$ for the first time}
        \State $Signatures \gets Signatures \cup \{sign\}$
        \If{$|Signatures| = |C|$}
            \State $Aggregate \gets aggregate(Signatures)$
            \For{$q \in Waiting$}
                \State Send $(COMPLETED, Aggregate)$ to $q$
            \EndFor
        \EndIf
    \EndUpon
\end{algorithmic}
\end{dianabox}

\begin{dianabox}{\textsc{QABQuorum}($\vin$, $proof$, phase parameter $\hat{f}$)}
\algoHead{Protocol for a quorum node $q$}
\begin{algorithmic}[1]
    
    \State $Acknowledged \gets \emptyset$
    \UponSimple{Receiving valid $(COMPLETED, sign)$ from relayer $r$ of committee $C$ for the first time}
        \State $Acknowledged \gets Acknowledged \cup C$
        \If{$|Acknowledged| \geq n - c_b \cdot \hat{f}$}
            \For{party $p \in \mathcal{P} \setminus Acknowledged$}
                \State Send $(DECIDE, \vin, proof)$ to $p$
            \EndFor
            \State We successfully disseminated our value
            \State \Return 
        \EndIf
    \EndUpon
    \Statex

    \For{Every relayer $r$}
        \State Send $(DISPERSE, \vin, proof)$ to $r$
    \EndFor

\end{algorithmic}
\end{dianabox}

We now give multiple properties of this protocol:

\begin{lemma}\label{lemma:quorum-validity}
    An honest party can only decide $\vin$.
\end{lemma}

\begin{proof}
    An honest party only decides a value $v$ if it receives a proof for it along with it. Because a QAB protocol assumes that a proof can only exist for $\vin$, an honest party can only decide $\vin$.
\end{proof}

\begin{lemma}\label{lemma:messages-party-relayer}
    At most $\mathcal{O}(n \log n)$ messages are exchanged between parties and relayers in a phase.
\end{lemma}

\begin{proof}
    A relayer sends at most one message to a party in its committee and a party answers at most once to such message for each committee it is part of. Because committees were made using \cref{thm:nodes2committees}, a party is part of at most $\mathcal{O}(\log n)$ committees. There are $n$ parties hence the $\mathcal{O}(n \log n)$ message bound.
\end{proof}

\begin{lemma}\label{lemma:messages-relayers-quorum}
    At most $\mathcal{O}(t \cdot \hat{f} \cdot \log n)$ messages are exchanged between relayers and quorum nodes.
\end{lemma}

\begin{proof}
    Each of the $\mathcal{O}(t)$ quorum node sends at most one message to the $\mathcal{O}(\hat{f} \log n)$ relayers, which answer at most once after they get an aggregate signature. Therefore, the total message complexity is $\mathcal{O}(t \cdot \hat{f} \cdot \log n)$.
\end{proof}

\begin{lemma}\label{lemma:messages-quorum-parties}
    At most $\mathcal{O}(t \cdot \hat{f})$ messages are exchanged between quorum nodes and parties.
\end{lemma}

\begin{proof}
    A quorum node sends a message directly to a party only when $|Acknowledged| \geq n - c_b \hat{f}$, and it does this only once per party missing in $Acknowledged$. Therefore, it sends at most $c_b \hat{f} = \mathcal{O}(\hat{f})$ messages to parties. Because there are $3t+1$ quorum nodes, the total message complexity is $\mathcal{O}(t \cdot \hat{f})$.
\end{proof}

\begin{lemma}\label{lemma:qab-phase-message-complexity}
    A whole QAB phase has message complexity $\mathcal{O}((n + t \cdot \hat{f}) \cdot \log n)$.
\end{lemma}

\begin{proof}
    We consider all kinds of messages sent by honest parties:
    \begin{enumerate}
        \item Messages exchanged between parties and relayers, which is $\mathcal{O}(n \log n)$ using \cref{lemma:messages-party-relayer}.
        \item Messages exchanged between relayers and quorum nodes, which is $\mathcal{O}(t \cdot \hat{f} \cdot \log n)$ using \cref{lemma:messages-relayers-quorum}.
        \item Messages sent by quorum nodes to parties, which is at most $\mathcal{O}(t \cdot \hat{f})$ using \cref{lemma:messages-quorum-parties}.
    \end{enumerate}
    Putting all of it together, we get a total message complexity of $\mathcal{O}((n + t \cdot \hat{f}) \cdot \log n)$.
\end{proof}

\begin{lemma}\label{lemma:qab-phase-no-honest}
    If no honest quorum node runs \textsc{QABQuorum}, then this phase has message complexity $\mathcal{O}((n + f \cdot \hat{f} )\cdot \log n)$.
\end{lemma}

\begin{proof}
    The $\mathcal{O}(n \log n)$ messages come from \cref{lemma:messages-party-relayer}. We note that if no honest quorum node starts \textsc{QABQuorum}, then any message sent by a quorum node is sent by a byzantine node. Only relayers react to messages sent by quorum nodes and send back a message at most once. As such, the up to $f$ byzantine quorum nodes can only trigger an additional $\mathcal{O}(f \cdot \hat{f} \cdot \log n)$ messages being sent by honest relayers. By summing these two kinds of messages, we get the total complexity. 
\end{proof}

\begin{lemma}\label{lemma:quorum-disseminated}
    If an honest quorum node reports having successfully disseminated its value, then all parties will have decided within one round.
\end{lemma}

\begin{proof}
    If a quorum node $q$ reports it has successfully disseminated its value, it means it sent its value to all parties with were not included in $Acknowledged$. Because a party $p$ gets added to $Acknowledged$ only if the leader receives an aggregate which $p$ must have partially signed, and $p$ is guaranteed to have decided when sending this message, this guarantees all honest parties in $Acknowledged$ have already decided. 

    We remark that $q$ then sends a decide message to all parties not in $Acknowledged$, which are then guaranteed to decide the value when receiving it if they have not done so already. It takes one round for these parties to receive the message, which proves the lemma.
\end{proof}

\begin{lemma}\label{lemma:quorum-honest}
    If an honest quorum node starts \textsc{QABQuorum} and $f \leq \hat{f}$, then it will receive acknowledgments from at least $n - c_b \cdot \hat{f}$ parties and report having successfully disseminated its value within $4$ rounds.
\end{lemma}

\begin{proof}
    If an honest quorum node starts \textsc{QABQuorum}, the following will happen in each round (or already have happened if another party, being honest or byzantine sent a valid message before):
    \begin{enumerate}
        \item 1st round: The quorum node sends a message to every relayer
        \item 2nd round: Each relayer relays the message to every party in its committee.
        \item 3rd round: Every honest party receives the value, decides it if not already done, signs an acknowledgment and sends it to its relayer.
        \item 4th round: Every relayer which receives a partial signature from all of its committee members aggregate them and send the aggregate signature to the quorum node.
        \item End of 4th round: The quorum node receives all aggregate signatures.
    \end{enumerate}
    We note that if $f \leq \hat{f}$, because committees are formed using \cref{thm:nodes2committees}, at most $c_b \cdot \hat{f}$ parties will be blocked. If a party is not blocked, it means it is part of a committee entirely made of honest parties, so the relayer will be able to acquire an aggregate signature and send it to the quorum node within $4$ rounds. Therefore, after $4$ rounds, in the worst case, the quorum node will be missing the signature of the up to $c_b \cdot \hat{f}$ blocked parties (we note that by definition, byzantine parties are always blocked). Therefore we will have $|Acknowledged| \geq n - c_b\hat{f}$ so the quorum node will send its value to the remaining parties and report having successfully disseminated its value.
\end{proof}

\section{Asynchronous Byzantine Agreement}
As a warm-up, we begin with our result for asynchronous Byzantine agreement, which is conceptually the simplest and already illustrates the main ingredients of our approach.

For the quorum agreement in asynchrony, we use a probabilistic protocol from Cachin et al.~\cite{JC:CacKurSho05} which has quadratic communication complexity, resulting in $O(t^2)$ communication for us. 

\subsection{Asynchronous QAB}
\label{sec:async qab}
The asynchronous QAB which we call \textsc{QABAsync} consists of a single QAB phase with parameter $t$, its implementation of \textsc{StartQuorum} is given here:

\begin{dianabox}{\textsc{StartQuorumAsync}$(\vin, \mathit{proof})$}
\algoHead{Implementation of \textsc{StartQuorum} in the asynchronous setting for a quorum node}
\begin{algorithmic}[1]
    \State Run $\textsc{QABQuorum}(\vin, \mathit{proof}, t)$
\end{algorithmic}
\end{dianabox}

\begin{theorem}
    \label{thm:qab async}
    \textsc{QABAsync} protocol solves QAB in asynchrony in the presence of at most $t$ Byzantine faults with communication complexity of $\mathcal{O}((n + t^2)\cdot \log n)$ and round complexity $\mathcal{O}(1)$. 
\end{theorem}
\begin{proof}
    We need to prove round and communication complexity, as well as Complete Correctness.

    \textit{Round complexity.} Because there are $3t+1$ quorum nodes, at least $2t+1 \geq 1$ quorum node will be honest. Therefore, as a consequence of \cref{lemma:quorum-disseminated,lemma:quorum-honest}, the round complexity is $\mathcal{O}(1)$.

    \textit{Communication complexity.} This is a direct consequence of \cref{lemma:qab-phase-message-complexity} with $\hat{f} = t$.
    
    \textit{Complete correctness} Termination is already guaranteed by the round complexity. \cref{lemma:quorum-validity} guarantees all parties decide $\vin$.
\end{proof}

\subsection{Complete Algorithm}
We are now ready to prove Theorem \ref{thm:async main}.
\mainThmAsync*
\begin{proof}
    We start by describing an algorithm that implements BA in asynchrony, and then we show its correctness, round complexity, and message complexity. Our algorithm uses an asynchronous Byzantine Agreement from Cachin et al.~\cite{JC:CacKurSho05}, which we refer to as $C$, and asynchronous QAB. The algorithm goes as follows.

    We use \cref{theo:decoupling} with the asynchronous BA algorithm from Cachin et al.~\cite{JC:CacKurSho05}, which has expected round complexity $\mathcal{O}(1)$ and expected message complexity $\mathcal{O}(n^2)$. We note that this protocol does not return a proof for the value decided. This can be simply worked around by having each party broadcast a $(n, t+1)$-partial signature for the value $v$ they decided, then wait for $t+1$ of these signatures to make a threshold signature which serves as a proof. This takes $1$ round and $n^2$ message complexity. Because this threshold signature requires the partial signature of at least one honest party, which only do it for $v$, this indeed works as a proof. 
    
    For the QAB protocol, we use the one described above which has $\mathcal{O}((n + t^2)\cdot \log n)$ message complexity and $\mathcal{O}(f)$ round complexity. Using \cref{theo:decoupling}, we get a protocol for BA with resiliency $t < n/3$, expected round complexity $\mathcal{O}(1) + \mathcal{O}(1) = \mathcal{O}(1)$ and expected message complexity $\mathcal{O}(t^2) + \mathcal{O}((n + t^2)\cdot \log n) = \mathcal{O}((n + t^2)\cdot \log n)$.
\end{proof}
\subsection{Infeasibility in Asynchrony}
\label{sec:lb}
On the infeasibility side, in the asynchronous setting, we show that any algorithm must incur a communication complexity of at least $\Omega(t^2)$. This leaves essentially no room for adaptivity in asynchronous algorithms and makes the algorithm described earlier in the section nearly optimal.

\thmLB*
We give a proof for this theorem in Appendix~\ref{sec:appendix:async lb}.

\section{Adaptive Byzantine Agreement in Synchrony}

\subsection{Quorum Agreement}
\label{sec:sync nf protocol}

We first start by describing a retrieval protocol for the synchronous setting with optimal resiliency $t < n/2$. The main idea for this retrieval protocol is that parties send a partial signature of their input to the leader. If the leader is not able to form a $t+1$ threshold signature for either $0$ or $1$, then it knows that both $0$ and $1$ can be decided (although it cannot prove it to other parties). In this case the leader sets its own value to $\bot$ and sends partial signatures for both $0$ and $1$ in the following rounds. We show that using this approach a proper proof can always be obtained after having $f+1$ distinct honest leaders.

\begin{dianabox}{\textsc{RetrievalLeaderSync}}
\algoHead{Retrieval protocol for the leader in synchrony}
\begin{algorithmic}
    \State Wait for $2\Delta$ units of time
    \If{there exists $\textbf{value} \in \{0,1\}$ which received at least $t+1$ partial signatures $(\rho_p)$}
        \State $\textbf{proof} \gets tcombine(\textbf{value}, (\rho_p))$
        \State \Return $(\textbf{value}, \textbf{proof})$
    \Else{}
        \State $\vin \gets \bot$
        \State \Return $\bot$
    \EndIf

\end{algorithmic}
\end{dianabox}

\begin{dianabox}{\textsc{RetrievalPartySync}($\vin$)}
\algoHead{Retrieval protocol for parties in synchrony}
\begin{algorithmic}
    \If{$\vin \ne \bot$}
        \State Send $(\vin, tsign(\vin))$ to the leader
    \Else{}
        \State Send both $(0, tsign(0))$ and $(1, tsign(1))$ to the leader
    \EndIf

\end{algorithmic}
\end{dianabox}

\begin{lemma}\label{lemma:bot-is-good}
    If an honest party sets its input to $\bot$, then both $0$ and $1$ can be decided according to strong unanimity.
\end{lemma}

\begin{proof}
    We first note that the input value $\vin$ of an honest party can only be changed to $\bot$ in one way. Therefore, if an honest party has input $\vin \in \{0,1\}$, it will always send a partial signature for $\vin$ when running the retrieval protocol.

    We consider the case where a party sets its input to $\bot$, it means it could not get a $t+1$ signatures for either $0$ or $1$. Using the previous part, it implies that at most $t$ honest parties had $0$ as input and that the same holds for $1$. Because $n \geq 2t+1$, we have $n-t > t$ so it implies that not all honest parties had the same input (which can only be $0$ or $1$). Therefore, according to strong unanimity, any value can be decided.
\end{proof}

\begin{lemma}\label{lemma:sync-retrieval-proof-good}
    If $\textsc{RetrievalLeaderSync}$ returns a non-$\bot$ value, then this value satisfies strong unanimity and the proof can be used to check it.
\end{lemma}

\begin{proof}
 Assume that the retrieval protocol returns $(v, proof)$. The proof shows that $t+1$ parties sent a partial signature to the leader with the value $v$. This implies that an honest party sent a partial signature for the value $v$. This can only happen if its input $v_{in}$ is $v$, in which case $v$ satisfies strong unanimity, or if it is $\bot$, in which case using \cref{lemma:bot-is-good}, $v$ also satisfies strong unanimity.
\end{proof}

\begin{lemma}\label{lemma:sync-retrieval-f1-proof}
    After the retrieval protocol is called by $f+1$ distinct honest leaders, it will return a non-$\bot$ value at least once.
\end{lemma}

\begin{proof}
    Let $nb_{0}$ (resp. $nb_1$) be the number of partial signatures that an honest leader would receive for $0$ (resp. $1$) from honest parties. Because every honest party starts by sending a threshold signature for either $0$ or $1$, at the beginning, we have $nb_{0} + nb_{1} = n-f$. If an honest party becomes leader for the first time and returns $\bot$, it will start sending partial signatures for both $0$ and $1$ instead of only its input value in subsequent calls. Thus $nb_{0} + nb_{1}$ will increment.

    We now look at the state of honest parties after $f$ distinct honest leaders have been using the retrieval protocol. If any of them returned a non-$\bot$ value, we are done. Otherwise, because $nb_{0} + nb_{1}$ starts at $n-f$ and increases by $1$ every time a new honest leader returns $\bot$, it means that by the time the ($f+1$)-th honest leader run the protocol, we will have $nb_{0} + nb_{1} \geq n-f+f = n \geq 2t+1$. Just looking at honest partial signatures, we will have at least $2t+1$ such signatures for either $0$ or $1$, so one of the two will have at least $t+1$ signatures and we will be able to get a proof and return a value.
\end{proof}

\begin{theorem}
    The pair $(\textsc{RetrievalLeaderSync}, \textsc{RetrievalPartySync})$ is a valid retrieval protocol in synchrony with resiliency $t < n/2$.
\end{theorem}

\begin{proof}
    Every party only sends one or two messages, so the message complexity is $\mathcal{O}(n)$. If the leader is honest, the protocol will return within exactly $2$ rounds. Using \cref{lemma:sync-retrieval-proof-good}, if a non-$\bot$ value is returned, then it satisfies strong unanimity and the proof certifies it. Using \cref{lemma:sync-retrieval-f1-proof}, when called by $f+1$ distinct honest leaders, it will not always return $\bot$. Therefore, it is a valid retrieval protocol.
\end{proof}

Using this retrieval protocol with our view-based protocol, as a consequence of \cref{thm:view-ba-main}, we obtain the following result:

\begin{theorem}\label{theo:view-quorum-algo}
In the synchronous setting, with $t < n/2$ and using the retrieval protocol $(\textsc{RetrievalLeaderSync}, \textsc{RetrievalPartySync})$, \textsc{ViewByzantineAgreement} achieves agreement and validity. Moreover, if $f < \lfloor \frac{n - t -1}{2} \rfloor$, it also satisfies termination within $\mathcal{O}(f)$ rounds and uses $\mathcal{O}(n\cdot (f + 1))$ messages.
\end{theorem}

We remark that if $f \leq \lfloor \frac{n - t -1}{2} \rfloor \approx n/4$, this protocol satisfies byzantine agreement. However, in the other case, some parties may never be able to decide because getting a $k$-threshold certificate is impossible if byzantine parties do not cooperate. We remark that in this latter case, $f = \Omega(n)$, and use the same technique as \cite{cohen2023make}, which is to detect this case and use a fallback protocol $\mathcal{A}_{fallback}$ with quadratic complexity $\mathcal{O}(n^2)$. The fallback protocol used is the one by Ren and Momose \cite{momosequadraticba} which has quadratic message complexity, $\mathcal{O}(n)$ round complexity, as proven in \cref{theo:momose-good} and optimal resiliency $t < n/2$. We also simplify the fallback detection approach used by Cohen et al. While their protocol may cause honest parties to still send messages an arbitrary amount of time after the start of the algorithm, our approach guarantees that no message will be sent by honest parties after $\mathcal{O}(n)$ rounds, so this protocol does not have to run infinitely. Moreover, while Cohen et al. have to run the the fallback protocol with a message delay $\delta$ twice as big as the original one, our approach does not have this issue. We give the full pseudocode and proof in \cref{sec:sync-proto-ba}.

\syncNf*

\subsection{QAB in Synchrony}
\label{sec:sync qab}
In this section, we describe and analyze our synchronous QAB protocol, which we call \textsc{QABSync}. It relies on $\mathcal{O}(\log t)$ QAB phases and a specific implementation of \textsc{StartQuorum} called \textsc{StartQuorumSync}.

\paragraph*{Protocol Description}
In synchronous Quorum to All Broadcast, we aim for the $O(n + tf)$ communication complexity and $O(f)$ rounds latency. However, the challenge is that processes do not know the value of $f$. To account for this, we do an \emph{exponential search} for $f$. That is, we try to solve the problem assuming $f$ is no more than $1$, then if it fails, we assume $f$ is no more than $2$, then $4$, $8$, and so on up to $2^{\lceil \log t \rceil}$. We will denote the current guess on $f$ with $\hat{f}$. The synchronous QAB then consists of launching QAB phase for every guess $\hat{f}$ until one successfully terminates:

\begin{dianabox}{\textsc{StartQuorumSync}$(\vin, proof)$}
\algoHead{Implementation of \textsc{StartQuorum} in the synchronous setting for a quorum node}
\begin{algorithmic}[1]
    \For{$i \in \{0,1,\ldots, \lceil \log t \rceil \}$}
        \State $\hat{f} \gets 2^i$
        \State Run $\textsc{QABQuorum}(\vin, proof, \hat{f})$ for $4$ rounds
        \If{the quorum node successfully disseminated its value}
            \State \Return
        \EndIf
    \EndFor
\end{algorithmic}
\end{dianabox}


\paragraph*{Protocol Analysis}

\begin{theorem}
    \label{thm:qab sync}
    \textsc{QABSync} solves QAB in synchrony in the presence of at most $t$ Byzantine faults with communication complexity of $\mathcal{O}((n \log t  + t \cdot f) \cdot \log n)$ and round complexity $\mathcal{O}(\log f)$. 
\end{theorem}
\begin{proof}
    \textit{Round complexity.} Let $\hat{f} = 2^{\lceil \log f \rceil}$, we have $\hat{f} \geq f$, so as a consequence of \cref{lemma:quorum-honest}, if an honest quorum runs the QAB phase with parameter $\hat{f}$, it will report having successfully disseminated its value within $4$ rounds and therefore won't run any QAB phase with higher parameter. Because there is at least one honest quorum node (the quorum size is $3t+1$), it will therefore report having successfully disseminated its value by after at most $\lceil \log f \rceil + 1$ QAB phases, so using \cref{lemma:quorum-disseminated}, all honest parties will have decided within round $4(\lceil \log f \rceil + 1) + 1$. So the round complexity is $\mathcal{O}(\log f)$.
    We need to prove round and communication complexity, as well as Complete Correctness.

    \textit{Communication complexity.} We saw that all honest quorum nodes stop the protocol by phase $2^{\lceil \log f \rceil}$, each of these phases has message complexity $\mathcal{O}((n+t\cdot\hat{f}) \cdot \log n)$ with $\hat{f}$ being the phase parameter according to \cref{lemma:qab-phase-message-complexity}, so the overall message complexity of this part is:
    \begin{align*}
        \sum_{i=0}^{\lceil \log f \rceil} \mathcal{O}((n+t\cdot 2^i) \cdot \log n) = \mathcal{O}((n \cdot \log f + t \cdot f) \cdot \log n)
    \end{align*}
    We note that byzantine parties may still trigger messages to be sent by honest parties in the upper phases, but using \cref{lemma:qab-phase-no-honest}, at most $\mathcal{O}((n + f \cdot \hat{f}) \log n)$ such messages are sent for phase parameter $\hat{f}$, summing over all possible phases, we get:
    \begin{align*}
        \sum_{i=0}^{\lceil \log t \rceil} \mathcal{O}((n + f \cdot 2^i) \log n) = \mathcal{O}((n \cdot \log t + f \cdot t) \cdot \log n)
    \end{align*}
    Putting these two together, we get the expected message complexity.
    
    \textit{Complete correctness} Termination is already guaranteed by the round complexity. \cref{lemma:quorum-validity} guarantees all parties decide $\vin$.
\end{proof}

\subsection{Complete Algorithm}
We can now compose our BA protocol and QAB to obtain an almost-optimal adaptive consensus. 
\mainThmSync*
\begin{proof}
    First, we remark that if $t \geq n/3$ (or more generally if $t = \Omega(n)$), then the protocol from \cref{thm:sync nf} has the correct message and round complexity.

    We can now assume that $t < n/3$. We will apply \cref{theo:decoupling}. We use the BA protocol from \cref{theo:gst-main-ba} which has resiliency $t < n/3$, message complexity $\mathcal{O}(n \cdot f)$ and round complexity $\mathcal{O}(f)$ (we note that this is a protocol for the partially synchronous setting, but as a consequence, it also works in the synchronous setting). We use the QAB described above from \cref{thm:qab sync} which has message complexity $\mathcal{O}((n \log t  + t \cdot f) \cdot \log n)$ (regardless of when quorum nodes get their value) and round complexity $\mathcal{O}(\log f)$. With \cref{theo:decoupling}, we get a protocol for BA with message complexity $\mathcal{O}(n \cdot f + (n \log t  + t \cdot f) \cdot \log n) = \mathcal{O}((n \log t  + t \cdot f) \cdot \log n)$ and round complexity $\mathcal{O}(f + \log f) = \mathcal{O}(f)$.
\end{proof}

\section{Adaptive Byzantine Agreement in Partial Synchrony}
\label{sec:agreement gst}

\subsection{Quorum Agreement}
We now describe a protocol to achieve $\mathcal{O}(n\cdot (f+1))$ message complexity in the partially synchronous setting. Our protocol uses the same approach as in the synchronous setting. We note that, because we only support $t < n/3$ instead of $t < n/2$, our protocol becomes much simpler. For example, the view-based approach only terminates if $f \leq \lfloor \frac{n - t - 1}{2} \rfloor$. However, if $t < n/3$, then $t \leq \lfloor \frac{n - t - 1}{2} \rfloor$, so the view-based approach always terminates and a fallback is not necessary.

We start with the retrieval protocol. This retrieval protocol is relatively easy: each party sends a ($t+1$)-threshold partial signature for their input value to the leader. Because $n \geq 3t+1$, the leader is guaranteed to be able to get a full threshold signature for either $0$ or $1$:

\begin{dianabox}{\textsc{RetrievalLeaderGST}}
\algoHead{Retrieval protocol for the leader in partial synchrony}
\begin{algorithmic}
    \State Wait for $n-t$ valid partial signatures
    \State Let $\textbf{value} \in \{0,1\}$ a value with at least $t+1$ partial signatures $(\rho_p)$
    \State $\textbf{proof} \gets tcombine(\textbf{value}, (\rho_p))$
    \Return $(\textbf{value}, \textbf{proof})$

\end{algorithmic}
\end{dianabox}

\begin{dianabox}{\textsc{RetrievalPartyGST}$(\vin)$}
\algoHead{Retrieval protocol for parties in partial synchrony}
\begin{algorithmic}
    \State Send $(\vin, tsign(\vin))$ to the leader

\end{algorithmic}
\end{dianabox}

\begin{theorem}
    The pair ($\textsc{RetrievalLeaderGST}$, $\textsc{RetrievalPartyGST}$) form a valid retrieval protocol in partial synchrony with resiliency $t < n/3$.
\end{theorem}

\begin{proof}
    We remark that $\textsc{RetrievalLeaderGST}$ always return a value-proof tuple. Moreover, the proof is a ($t+1$)-threshold signature of parties certifying they have the same input value $v$. This implies that at least one honest party has this value and therefore that $v$ can be decided.

    For the message and round complexity, we observe that its message complexity is $\mathcal{O}(n)$. Moreover, after GST, the leader will have to wait at most $1$ round before receiving $n-t$ inputs from honest parties.

    Finally, after $GST$, if the leader is honest, then it will receive $n-t$ suggestions. We remark that $3t + 1 \leq n$ so the leader will receive at least $2t+1$ partial signatures for either $0$ or $1$. As a consequence, either $0$ or $1$ will receive at least $t+1$ partial signatures and the leader will be able to make a proof.
\end{proof}

Using this retrieval protocol along with the view-based result from \cref{thm:view-ba-main}. Because $t < n/3$ implies $f \leq t \leq \lfloor \frac{n - t - 1}{2} \rfloor$, we immediately get an optimal protocol for the partially synchronous setting:

\thmGSTnf*

\subsection{QAB in Partial Synchrony}
\label{sec:qab partial sync}
In this section, we present an algorithm that solves QAB sending $\mathcal{O}((n + tf) \log t \cdot \log n)$ words after GST in the partially synchronous setting with round complexity $\mathcal{O}(f)$.

Our algorithm uses the QAB protocol for synchrony and adapts it to the partially synchronous setting by making two major modifications: running all phases \textbf{in parallel} and running them \textbf{slowly}:
\begin{itemize}
    \item \textbf{In parallel}: The serial approach used in synchrony does not work for the partially synchronous setting, as a phase $\hat{f}$ may fail to disseminate its value if ran before GST, even if $\hat{f} \geq f$. Instead, we run all $\log t$ phases in parallel. We note that the last phase has communication complexity $\Omega(t^2)$. In order not to reach this point, we use our second trick to ensure that higher phases never get to run in their entirety.
    \item \textbf{Slowly}: Instead of sending the value and proof to all relayers at once, we do it one at a time. To be more precise, every two rounds, we send the value and proof to a new relayer (if we sent the value to all relayers, we start again from the beginning). Similarly, instead of sending the value to the remaining $c_B \cdot \hat{f}$ parties all at once, we do it one party per round. Doing so, we can ensure that quorum nodes and relayers together only send $\mathcal{O}(t \log t)$ messages each round. By then showing that the protocol finishes within $\mathcal{O}(f \log n)$, we get immediately the message complexity as a bonus.
\end{itemize}

To do this, we run $\mathcal{O}(\log t)$ custom QAB phases, where \textsc{QABRelayer} is replaced \textsc{QABRelayerGST} and \textsc{QABQuorum} is replaced \textsc{QABQuorumGST}. Along with the implementation of \textsc{StartQuorum} given below, we call the resulting protocol \textsc{QABPartialSync}. These protocols are described in \cref{sec:apx:qab phase gst}. We note that this only changes the pace at which messages are exchanged between quorum nodes on one side and relayers and parties on the other side. As such, properties on a quorum phase which do not address message or round complexity, or are about the interaction between relayers and parties are still preserved. These protocols also rely on views, where we assume that each view lasts for at least $2\Delta$ units of time.

\begin{dianabox}{\textsc{StartQuorumGST}$(\vin, proof)$}
\algoHead{Implementation of \textsc{StartQuorum} in the partially synchronous setting for a quorum node}
\begin{algorithmic}[1]
    \InParallel{for $i = 0,1,\ldots, \lceil \log t \rceil$}
        \State $\hat{f} \gets 2^i$
        \State Run $\textsc{QABQuorumGST}(\vin, proof, \hat{f})$
    \EndParallel
    \State Stop all QAB quorum protocols once one of them is done disseminating its value
\end{algorithmic}
\end{dianabox}

\begin{lemma}\label{lemma:qab-gst-round-complexity}
    \textsc{QABPartialSync} protocol has $\mathcal{O}(f \log n)$ round complexity.
\end{lemma}

\begin{proof}
    The adaptive BA algorithm used has $\mathcal{O}(f)$ round complexity. Let $\hat{f} = 2^{\lceil \log f \rceil} \geq f$. As stated in \cref{lemma:quorum-honest}, after running the phase with parameter $\hat{f}$, the quorum node is missing the proof that at most $c_B \cdot \hat{f} = \mathcal{O}(f)$ parties received the value. After GST, to gather all aggregate signature for phase $\hat{f}$, because the proof is distributed to relayers one every two rounds and no guarantee can be made to attempts made before GST, the number of rounds taken is directly proportional to the number of relayers, which is $\mathcal{O}(\hat{f} \log n) = \mathcal{O}(f \log n)$. After getting all of these proofs, the quorum node knows that at most $c_B \cdot \hat{f} = \mathcal{O}(f)$ parties are missing a proof. Because it sends the proof to parties one at a time, this takes an additional $\mathcal{O}(f)$ rounds. Therefore, the phase with parameter $\hat{f}$ is guaranteed to be done within $\mathcal{O}(f \log n)$ after GST. We note that another phase may be done before, but in any case, if a phase is done for an honest quorum node, then using \cref{lemma:quorum-disseminated}, all honest parties will decide within $1$ round (or $1$ round after GST if it has not been reached yet, which would give a round complexity of 1).
\end{proof}

\begin{lemma}\label{lemma:qab-gst-round-msg}
    At most $\mathcal{O}(t \log t)$ messages are sent every view by quorum nodes and by relayers to quorum nodes.
\end{lemma}

\begin{proof}
    We look inside a phase $\hat{f}$. We remark that by design, every view, a single relayer will accept messages from quorum nodes (in a circular way) and sends at most one message per quorum node in the same view. Because this is the only kind of message relayers send to honest parties and the quorum has size $\mathcal{O}(t)$, this in total amounts to $t$ messages from relayers to quorum node. Looking at a single quorum node, every view, it sends one message to a relayer and two decide messages to parties, so their total message complexity is also $\mathcal{O}(t)$. When summing it over all $\mathcal{O}(\log t)$ phases, we get the expected $\mathcal{O}(t \log t)$ messages.
\end{proof}

\begin{lemma}\label{lemma:qab-gst-msg-complexity}
    \textsc{QABPartialSync} protocol has $\mathcal{O}((n + tf) \log t \cdot \log n + l \cdot t \cdot \log t)$ word complexity, where $l$ is the number of rounds after GST for the first quorum node to call $\textsc{StartQuorum}(\vin, proof)$.
\end{lemma}

\begin{proof}
    As stated in \cref{lemma:messages-party-relayer}, at most $\mathcal{O}(n \log n)$ messages are exchanged per phase between parties and relayers. There are $\mathcal{O}(\log t)$ phases which gives a $\mathcal{O}(n \log n \log t)$ message complexity for this part. All other messages are accounted by \cref{lemma:qab-gst-round-msg} which proves that at most $\mathcal{O}(t \log t)$ are sent every view. Because of \cref{lemma:qab-gst-round-complexity}, quorum nodes finish running their phase protocol within $\mathcal{O}(l + f \log n)$ rounds. Because each view consists of $\mathcal{O}(1)$ rounds, this gives in total a $\mathcal{O}(t \log t \cdot (l + f \log n))$ message complexity. We note that even after honest quorum nodes are done, byzantine quorum nodes can still ask relayers for an aggregate value. However the relayer will only respond once per quorum node, so the total overhead from this is $\mathcal{O}(\hat{f} \log n \cdot f)$ per phase so $\mathcal{O}(t \log n \cdot f)$ when summing over all phases. Summing all of these together, we get the expected round complexity.
\end{proof}

\begin{lemma}\label{lemma:qab-gst-valid}
    \textsc{QABPartialSync} protocol is a valid QAB protocol.
\end{lemma}
\begin{proof}
    \cref{lemma:qab-gst-round-complexity} covers termination while \cref{lemma:quorum-validity} covers validity of the QAB protocol.
\end{proof}

\subsection{Complete Algorithm}
\mainThmGST*
\begin{proof}
We will use \cref{theo:decoupling}. We consider our BA protocol from \cref{theo:gst-main-ba} which has round complexity $\mathcal{O}(f)$ and message complexity $\mathcal{O}(n \cdot f)$, along with the previous QAB protocol which according to \cref{lemma:qab-gst-round-complexity,lemma:qab-gst-msg-complexity,lemma:qab-gst-valid} has round complexity $\mathcal{O}(f \log n)$ and message complexity $\mathcal{O}((n + tf) \log t \cdot \log n + l \cdot t \cdot \log t)$ where $l$ is the number of rounds after GST for the first quorum node to call $\textsc{StartQuorum}(\vin, proof)$.

Using \cref{theo:decoupling}, we obtain a BA protocol with resiliency $t < n/3$, round complexity $\mathcal{O}(f) + \mathcal{O}(f \log n) = \mathcal{O}(f \log n)$, and message complexity $\mathcal{O}(tf) + \mathcal{O}((n + tf) \log t \cdot \log n + \mathcal{O}(f) \cdot t \cdot \log t) = \mathcal{O}((n + tf) \log t \cdot \log n)$.
\end{proof}

\section{Conclusion}

In this work, we study the communication and round complexity of deterministic binary Byzantine Agreement across all major timing models: synchronous, partially synchronous, and asynchronous networks.

Our main contribution is twofold. First, we present optimal adaptive algorithms that simultaneously achieve optimal communication and round complexity with optimal resilience. Second, for the high-scale regime where the total number of nodes significantly exceeds the tolerable faults, we achieve near-optimal communication and round complexity.

A key technical ingredient is a deterministic committee assignment scheme obtained via bipartite dispersers. This approach enables efficient information dissemination while remaining robust against adaptive adversaries.

Our work opens natural directions for future research. One is to remove the polylogarithmic factors and achieve truly optimal complexity, which likely requires fundamentally new techniques beyond current disperser constructions. Another is to extend our results to multi-valued Byzantine Agreement. Finally, regarding partial synchrony, in this paper, we assume the presence of an independent view synchronization protocol. This can be achieved by relying on specific clock assumptions. Doing so without any assumption on clocks remains an open question.

\clearpage
\newpage

\begingroup
\emergencystretch=2em
\bibliographystyle{ACM-Reference-Format}

\bibliography{bib/refs,bib/abbrev3,bib/crypto_crossref}
\endgroup

\newpage
\appendix 
\section*{Appendix}

\section{Lower Bound in Asynchrony}
\label{sec:appendix:async lb}
In this section, we prove the lower bound theorem, which implies that no adaptive communication complexity is possible in an asynchronous setting.
\thmLB*
\begin{proof}
    In this proof, by the \emph{setup} we mean a combination of identities of faulty parties, proposals of honest parties, behavior of faulty parties and the schedule of messages. Let $M$ be a random variable denoting the number of messages sent by the protocol.

    We will prove that there exists a setup \texttt{Main} in which no party is byzantine and $\mathbb{E}(M) \geq t^2/4$.
    
    \texttt{Main} is defined as follows. We partition the parties $\mathcal{P}$ into two sets $B$ and $C$ such that $|C| = t/2$ and $|B| = n - t/2$.
    \begin{itemize}
        \item All parties in $\mathcal{P}$ are honest and have input $0$.
        \item Messages between parties in $C$ are delayed until after everyone decides.
        \item For every party $p \in C$, the first $t/2$ messages sent to $p$ from processes in $B$ are delayed until after everyone decides.
        \item The other messages are delivered immediately.
    \end{itemize}

    We want to show that this scheduling is valid for the asynchronous model. I.e., we want to show that almost surely, all parties will decide before any of the delayed messages are received.

    We first prove this for every party in $B$. To do that, we consider the setup $S_B$ for which we'll show that (i) all processes in $B$ must decide and (ii) processes in $B$ can not distinguish $S_B$ from \texttt{Main}. $S_B$ is defined as follows:
    \begin{itemize}
        \item All parties in $B$ are honest and have input $0$.
        \item All parties in $C$ are byzantine. They act as if they were executing $\mathcal{A}$ with input $0$, but they ignore all messages from other parties in $C$ and the first $t/2$ messages from parties in $B$.
        \item All messages are immediately delivered.
    \end{itemize}
    Clearly, parties in $B$ cannot distinguish $S_B$ and \texttt{Main}. Moreover, in $S_B$, all messages are eventually delivered, and less than $t$ parties are byzantine. Therefore, $\mathcal{A}$ must ensure probabilistic termination, so parties in $B$ will eventually decide.

    Now, let us prove that for every party $p$ in $C$ it eventually decides in \texttt{Main}. To do so, we introduce a family of setups $\{S_Q\}$, for all $Q \subseteq B$ such that $|Q| = t/2$. $S_Q$ is defined as follows:
    \begin{itemize}
        \item All parties in $B \setminus Q$ are honest and have input $0$.
        \item $p$ is honest and has input $0$.
        \item All parties in $C \setminus \{p\}$ are byzantine. They act as if they were executing $\mathcal{A}$ with input $0$ but ignore messages from other parties in $C$ as well as the first $t/2$ messages from parties in $B$. Moreover, they do not send any message to $p$.
        \item All parties in $Q$ are byzantine. They act as if they were executing $\mathcal{A}$ with input $0$ but do not send the first $t/2$ messages to $p$.
        \item All messages are immediately delivered.
    \end{itemize}
    In setting $S_Q$, $t/2 + t/2 - 1 < t$ parties are byzantine and all messages are eventually delivered. Therefore, $A$ must ensure probabilistic termination so all honest parties, and in particular $p$, will eventually decide.
    
    For any execution of $\mathcal{A}$ in \texttt{Main}, there exists a $Q \subseteq B$ such that $|Q| = t/2$ and $p$ cannot distinguish between \texttt{Main} and $S_Q$. Since there is only a finite number of such setups $S_Q$ and in any of these setups $p$ eventually decides, it implies that $p$ will eventually decide in \texttt{Main}. Therefore, \texttt{Main} is covered by protocol $\mathcal{A}$, so the behavior of all parties should satisfy probabilistic termination, agreement, and validity. Using validity and the fact that every party has input $0$, we obtain that every party will eventually decide $0$ in \texttt{Main}.

    We will now show that if for \texttt{Main} the expected number of messages sent is less than $t^2/4$, then some party in $C$ decides $0$ without receiving any message with a non-zero probability.
    
    Having $\mathbb{E}(M) < t^2/4$ and using Markov's inequality, we get that:
    \begin{align*}
        \mathbb{P}(M \geq t^2/2) \leq 0.5
    \end{align*}
    For a party $p \in C$, let $E_p$ be the event that no more than $t/2$ messages are sent to $p$ from $B$. We remark that $\bigcap_{p\in C} \overline{E_p} \subset \{ M \geq t^2/2\}$. Using the previous inequality and by taking the complement, we get:
    \begin{align*}
        \mathbb{P}\left(\bigcup_{p \in C} E_p \right) &\geq 0.5 
    \end{align*}
    and by Union bound
    \begin{align*}
        \sum_{p \in C} \mathbb{P}(E_p) &\geq 0.5.
    \end{align*}
    Therefore, there exists $p \in C$ such that $\mathbb{P}(E_p) \geq 1/(2|C|) = 1/t$.

    Because the first $t/2$ messages are delayed after the decision, under the assumption that $p$ decides, which happens almost surely, this means that $p$ will decide its output value with probability at least $1/t > 0$ without receiving any message. Conditioning on this event and because deciding happens almost surely, this implies that there exists $T > 0$ such that $p$ will decide a value within time $T$ without receiving any message with probability at least $1/t > 0$.

    We derive a contradiction from this last fact. We consider the following \texttt{Final} setup:
    \begin{itemize}
        \item All parties are honest.
        \item Parties in $\mathcal{P}\setminus \{p\}$ have input $1$.
        \item Party $p$ has input $0$.
        \item All messages are delayed for the first $T$ units of time, afterwards they are received immediately.
    \end{itemize}
    We first consider \texttt{Final'} with the slight difference that $p$ is byzantine, but acting as honest with input $0$. In this case, because of validity, protocol $\mathcal{A}$ guarantees that parties in $\mathcal{P}\setminus \{p\}$ decide $1$. Moreover, these parties cannot distinguish \texttt{Final} and \texttt{Final'}, so they will also decide $1$ in \texttt{Final}. However, for the first $T$ units of time, $p$ cannot distinguish \texttt{Final} and \texttt{Main}. Therefore, it will decide $0$ with probability at least $1/t > 0$. However, all parties are honest and this breaks the agreement with a non-zero probability (one honest party decides $0$, the others decide $1$), hence the contradiction.
\end{proof}

\section{Quorum Byzantine Agreement}
\label{sec:app-quorum-ba}
In this section, we give a pseudocode as well as formal proofs for out partially synchronous agreement algorithm.

\subsection{View-Based Protocol}
\label{sec:apx:alg gst pseudocode}
This section provides a pseudocode for view-based byzantine agreement protocol. The protocol uses the threshold value $k = \lceil \frac{n+t+1}{2} \rceil$ for the threshold signatures.

In the protocol, for every message sent, we implicitly bundle along it the view number when it was sent. If a message is received by a party with a view number different than the current one, it is ignored. This prevents the adversary from using messages from previous views in case a party becomes leader multiple times. The only exceptions are the \textsc{SendCommit} message which is always accepted if the commit is valid and $\textsc{Suggest}((COMMIT,-))$ message which a party will always save and use the next time it becomes a leader.

\begin{dianabox}{\textsc{ViewByzantineAgreement}($\vin$)}
\algoHead{Protocol for Byzantine Agreement in partial synchrony}
\begin{algorithmic}
\State $\textbf{key} \gets \bot$
\State $\textbf{lock} \gets \bot$
\State $\textbf{commit} \gets \bot$

\Statex{$\triangleright$  Decide as soon as we get a commit value}
\UponTrue{Receiving valid $\textsc{SendCommit}(\textbf{value}, \textbf{proof})$}
    \If{$\textbf{commit} = \bot$}
        \State $\textbf{commit} \gets (\textbf{value}, \textbf{proof})$
        \State Decide \textbf{value}
    \EndIf
\EndUpon
        
\Statex{$\triangleright$  Run the view-based protocol}
\For{view number $view \gets 0,1, \ldots$}
    \If{This party is the current leader}
        \If{$\textbf{commit} = \bot$}
            \State Run ViewLeaderProtocol() in parallel for duration $11\Delta$
        \Else{}
            \UponSimple{Receiving $(\textsc{Complain}$) from a party $v$ for the first time}
                \State Send $\textsc{SendCommit}(\textbf{commit})$ to $v$
            \EndUpon
        \EndIf
    \EndIf
    \State Run ViewPartyProtocol($\vin$, $view$, \textbf{key}, \textbf{lock}, \textbf{commit}) for duration $11\Delta$
\EndFor

\end{algorithmic}
\end{dianabox}

\begin{dianabox}{\textsc{ViewLeaderProtocol}}
\algoHead{View-based protocol part exclusive to the leader}

\begin{algorithmic}[1]
        \Statex{$\triangleright$  Choose which value to propose}
        \State Broadcast \textsc{RequestSuggestion}
        \State Wait for valid \textsc{Suggest}(m) from $k$ parties
        \If{one of the value is a commit $C$ with value $v$:}
            \State $\textbf{value} \gets v$
            \State $\textbf{commit} \gets C$
            \State Jump to broadcasting the commit value
        \ElsIf{One of the value is a key value:}
            \State Let $key$ be the key with the highest view number and $v$ its value
            \State $\textbf{value} \gets v$
            \State $\textbf{prop} = (KEY, key)$
        \Else{}
            \State Broadcast \textsc{RunRetrieval}
            \State $(\textbf{value}, \textbf{proof}) \gets \textsc{RetrievalLeader}()$
            \If{$\textbf{value} = \bot$}
                \State No value to propose
                \State Stay silent for the rest of the round
            \EndIf
            \State $\textbf{prop} = (COMBINE, proof)$
        \EndIf

        \Statex
        \Statex{$\triangleright$  Approve the proposed value and get a key}
        \State Send \textsc{ProposeKey}($\textbf{value}, \textbf{prop}$) to every party.
        \State Wait for valid \textsc{CheckedKey}($\rho_p$) from $k$ parties
        \State $\textbf{key\_proof} \gets tcombine((KEY, \textbf{value}), (\rho_p))$

        \Statex
        \Statex{$\triangleright$ Approve the key and get a lock}
        \State Send \textsc{ProposeLock}($\textbf{value}, \textbf{key\_proof}$) to every party.
        \State Wait for valid \textsc{CheckedLock}($\rho_p$) from $k$ parties
        \State $\textbf{lock\_proof} \gets tcombine((LOCK, \textbf{value}), (\rho_p))$

        \Statex
        \Statex{$\triangleright$ Approve the lock and get a commit}
        \State Send \textsc{ProposeCommit}($\textbf{value}, \textbf{lock\_proof}$) to every party.
        \State Wait for valid \textsc{CheckedCommit}($\rho_p$) from $k$ parties
        \State $\textbf{commit\_proof} \gets tcombine((COMMIT, \textbf{value}), (\rho_p))$

        \Statex
        \Statex{$\triangleright$ Broadcast the commit value}
        \State Send \textsc{SendCommit}($\textbf{value}, \textbf{commit\_proof}$) to every party.
\end{algorithmic}
\end{dianabox}

\begin{dianabox}{\textsc{ViewPartyProtocol}(input value $\vin$, view number $view$, \textbf{key}, \textbf{lock}, \textbf{commit})}
\algoHead{View-based protocol part exclusive to parties}

\begin{algorithmic}
        \Statex{$\triangleright$  Complain if you have not yet decided}
        \If{$\textbf{commit} = \bot$}
            \State Send $(\textsc{Complain})$ to the leader
         \EndIf

        \Statex
        \Statex{$\triangleright$  Choose which value to propose}
        \State Wait for a message \textsc{RequestSuggestion} from the leader
        \If{$\textbf{commit} \neq \bot$}
            \If{This party never sent the commit value to the leader}
                \State Send \textsc{Suggest}($(COMMIT, \textbf{commit})$) to the leader
            \EndIf
            \State Wait for the rest of the view
        \ElsIf{$\textbf{key} \neq \bot$}
             \State Send $\textsc{Suggest}((KEY, \textbf{key}))$ to the leader
        \Else
            \State Send $\textsc{Suggest}(())$ to the leader
        \EndIf

        \Statex
        \Statex{$\triangleright$  Run the retrieval procedure if needed}
        \Upon{Receiving from leader}{\textsc{RunRetrieval}}{}
            \State Run $\textsc{RetrievalParty}(v_{IN}$)
        \EndUpon
        
        \Statex{$\triangleright$  Check and sign the key}
        \State Wait for a valid \textsc{ProposeKey}($v$, ($type, proof$)) from the leader
        \If{$\textbf{lock} \ne \bot$}
            \If{$type = COMBINE$ or $type = KEY$ with a view number strictly less than the $\textbf{lock}$ view number}
                \State Wait for the rest of the view
            \EndIf
        \EndIf
        \State Send \textsc{CheckedKey}($tsign((KEY, v))$) to the leader
        
        \Statex
        \Statex{$\triangleright$  Check and sign the lock value}
        \State Wait for a valid \textsc{ProposeLock}($v$, $proof$) from the leader
        \State $\textbf{key} \gets (v, view, proof)$
        \State Send \textsc{CheckedLock}($tsign((LOCK, v))$) to the leader
        
        \Statex
        \Statex{$\triangleright$  Check and sign the commit value}
        \State Wait for a valid \textsc{ProposeCommit}($v$, $proof$) from the leader
        \State $\textbf{lock} \gets (v, view, proof)$
        \State Send \textsc{CheckedCommit}($tsign((COMMIT, v))$) to the leader
\end{algorithmic}
\end{dianabox}

\subsection{Proofs}
\label{sec:apx:gst alg proofs}

We will prove the correctness and time complexity of the protocol with a series of lemmas:
\begin{lemma}
    During a view, only a single value can get key, lock or commit proofs.
\end{lemma}

\begin{proof}
    A key, lock or commit is a threshold signature signed by $k = \lceil \frac{n+t+1}{2} \rceil$ parties. Therefore, if two values $v_1$ and $v_2$ get a key, lock or commit proof, this means that $k$ parties signed each one, so at least $n-2(n-k) = 2k-n$ parties signed both of them.
    We remark that $2k - n = 2\lceil \frac{n+t+1}{2} \rceil - n \geq n+t+1-n \geq t+1$ therefore at least $t+1$ parties signed $v_1$ and $v_2$. This implies that at least one honest party did, which implies that $v_1 = v_2$ because an honest party will only sign a single proof in a view.
\end{proof}

\begin{lemma}\label{lemma:klock-done}
    During a view, if $k-t$ honest parties acquire a lock proof for a value $v$, then any subsequent key, lock or commit proof will be for $v$.
\end{lemma}

\begin{proof}
    We assume that a set $Q$ of at least $k-t = \lceil \frac{n+t+1}{2} \rceil - t = \lceil \frac{n-t+1}{2} \rceil$ honest parties got a lock proof for a value $v_1$ during view $view\_1$.  We will show that all subsequent key proofs generated will be for $v_1$, which implies that only lock proofs and then commit proofs can be generated for $v_1$ thereafter, as it requires a key proof for it. Assume by contradiction that a key proof for a different key $v_2 \ne v_1$ is generated and consider the first view $view_2 \geq view_1$  generating such a proof. We have $view_2 > view_1$ because of the previous lemma. If we got a key proof for $v_2$, it means $k$ parties signed it. Because $|Q| + k \geq \lceil \frac{n-t+1}{2} \rceil + \lceil \frac{n+t+1}{2} \rceil  \geq n+1 > n$, it means that the set of parties in $Q$ and the ones which signed the key proof cannot be disjoint. So there exists a party $p \in Q$ which signed the key proof for $v_2$. By the minimality of $view_2$, $p$ still had the lock proof for $v_1$ at that point. Therefore, for $p$ to accept $v_2$, given that $v_2 \ne v_1$, this value must come with a key proof with view number $view_3 \geq view_1$. Therefore there is a key proof for $v_2$ at time $view_3$ with $view_1 \leq view_3 < view_2$. This contradicts the minimality of $view_2$, hence the result.
\end{proof}

\begin{lemma}\label{lemma:gst-agreement}
    During an entire execution, at most one value will get commit proofs.
\end{lemma}

\begin{proof}
    This is a direct consequence of \cref{lemma:klock-done}. For a commit proof to be created, it needs partial signatures from $k$ party and for an honest party to create a partial signature for a commit proof, it needs to receive a lock proof for the commit value during the same view. This implies that $k-t$ honest parties got a lock proof for this value in this view. Therefore, using \cref{lemma:klock-done}, all subsequent commit proofs will be for the same value.
\end{proof}

Because we only decide on a value which gets a commit proof, this implies agreement. We now prove the validity of this protocol:

\begin{lemma}\label{lemma:gst-proof-good}
    If a value has a key, lock or commit proof, then it satisfies strong unanimity.
\end{lemma}

\begin{proof}
    We note that if a value has a commit proof, then it must have a lock proof which then implies that it must have a key proof. So we can assume that a value $v$ has a key proof for it. For a key proof to be valid, it must come either from:
    \begin{itemize}
        \item A previous key proof, which comes from a previous key proof, we recursively look at this key proof until it does not come anymore from a key proof and get to the second case.
        \item A value along with its proof returned from the retrieval sub-protocol. We remark that by definition of the retrieval procedure, the proof must allow any party to locally check that the value satisfies strong unanimity.
    \end{itemize}
\end{proof}

\begin{corollary}\label{lemma:gst-validity}
    If a value gets decided, then it satisfies strong unanimity.
\end{corollary}

Finally, we can look at termination:
\begin{lemma} \label{lemma:gst-honest-commit}
    After GST, assuming $f \leq n-k$, if a view start with a leader who is honest then at the end of the view, either every honest party gets a commit proof or the leader ran the retrieval protocol and it failed (i.e returned $\bot$).
\end{lemma}

\begin{proof}
    After GST, messages are received withing $\Delta$ times. Because we assume $f \leq n-k$, this means that $k$ honest parties can always sign the value as long as it is valid, which will be the case if the leader is honest, and the condition for the key proof is passed. Because a view lasts $11 \Delta$ rounds, this is enough to make sure all messages get sent and received by honest parties before the view ends.

    If the leader starts the view with a commit proof, then it will not run the leader part of the protocol. However, we remark that in this case, every honest party which does not have a commit proof yet will send a \textsc{Complain} message to the leader. In turn, the leader will send its commit proof (at most once) to each party which sends a complain message. Therefore after two rounds, every honest party will have the commit proof.
    
    In the case where the leader does not have the commit proof at the beginning of the view, but receives at some point a valid \textsc{Suggest} message with a commit proof, then it will broadcast it to all parties which ends the proof in this case. Thus we can now assume the leader did not already have a commit proof and did not receive a \textsc{Suggest} message with a commit proof.

    Therefore, the only part left to prove is that all honest parties will sign the suggested value for the leader to get a key proof. Assume by contradiction that an honest party $p$ does not sign a suggested value. It means it has a lock proof and the suggested value has a $COMBINE$ type or comes from a lock proof with a strictly smaller view number. However, if a party has a lock proof for a value $v$ and a given view number $view$, this means that at least $k$ parties got a key proof for $v$ with view number $view$. So at least $k-t = \lceil \frac{n+t+1}{2} \rceil - t = \lceil \frac{n-t+1}{2} \rceil$ honest parties had at some point a key proof for $v$ with view number $view$. Because the leader waits for $k$ parties, he received the suggestion from at least $k$ of them. We remark that $k-t + k = \lceil \frac{n-t+1}{2} \rceil + \lceil \frac{n+t+1}{2} \rceil \geq n+1 > n$. Therefore, the leader received a suggestion from an honest party which had previously a key proof for $v$ with view number $view$. Because a party's key proof can only be replaced by one with a higher view number, this means that the leader got a suggestion from $q$ with a view number at least $view$. Therefore, because the leader is honest, according to the protocol, it will suggest a value with a key proof at least $view$, so $p$ would have accepted it, hence the contradiction.
\end{proof}

\begin{lemma}\label{lemma:gst-decide-2f}
    Assuming, $f \leq n-k$, after at most $2f+1$ views after GST, all honest parties will have decided a value.
\end{lemma}

\begin{proof}
    We consider the first $2f+1$ views after GST. Because $n \geq 2f+1$, each one of them has a different leader and at least $2f+1 \geq f+1$ have an honest leader. We remark by definition of the retrieval procedure that if it is called with $f+1$ different honest leaders, it will succeed at least once. Therefore, among these $f+1$ views with different honest leaders, there is one view where the leader either did not run the retrieval protocol or ran it and it succeeded. In both cases, using \cref{lemma:gst-honest-commit}, all honest parties will have a commit proof and thus have decided by the end of the view.
\end{proof}

\begin{corollary} \label{coro:gst-f-views}
    Assuming $f \leq n-k$, all honest parties will have decided within $\mathcal{O}(f+1)$ rounds.
\end{corollary}

We now cover the message complexity of this protocol:

\begin{lemma}\label{lemma:gst-view-messages}
    When ran for a single view, this protocol has message complexity $\mathcal{O}(n)$.
\end{lemma}

\begin{proof}
    This observation comes from the fact that at most $\mathcal{O}(n)$ messages are sent by honest parties during a view. Indeed, excluding the retrieval procedure, \textsc{ViewPartyProtocol} only sends at most $5$ messages to the leader for each party, and \textsc{ViewLeaderProtocol} only sends at most $6n$ messages to other parties. Moreover, a property of the retrieval protocol is that it only requires $\mathcal{O}(n)$ messages in total and it is only ran at most once per view. Therefore a view has message complexity $\mathcal{O}(n)$.
\end{proof}

From this previous lemma, we get the following corollary, which reveals itself useful when $f > n-k$, because in this case, $f = \Omega(n)$, so $\mathcal{O}(n^2) = \mathcal{O}(n \cdot f)$.
\begin{corollary}\label{coro:gst-bad-complexity}
    When ran for $n$ views, this protocol has message complexity $\mathcal{O}(n^2)$.
\end{corollary}

\begin{lemma}\label{lemma:gst-complexity}
    Assuming $f \leq n- k$, this protocol has message complexity $\mathcal{O}(n\cdot (f+1))$.
\end{lemma}

\begin{proof}

    We consider the total number of messages sent after GST. We divide the lifetime of the protocol in two:
    
        \textbf{During the first $2f+1$ views after GST}: using \cref{lemma:gst-view-messages}, each view has message complexity $\mathcal{O}(n)$, so the total complexity is $\mathcal{O}(n \cdot (f+1))$.
        
        \textbf{After the first $2f+1$ views after GST}: because we assume $f \leq n-k$, this implies using \cref{lemma:gst-decide-2f} that all honest parties have a commit value and already decided. We look at messages that can still be sent by honest parties. First, we remark that all $\textsc{Complain}$ messages at this point will come from byzantine parties, as an honest party complain only if it has no commit value, which they all have at that point. Because the leader will only respond at most once to a complain message from a given party $v$ for the entire lifetime of the protocol, the total message complexity from complain messages after $2f+1$ views is at most $f \cdot (n - f) = \mathcal{O}(n\cdot f)$. We remark that other than responding to complain messages, a leader which already has a commit value by the start of its view does not send any other message. Therefore, at this point, a leader sending any non-complain message implies that it is byzantine. Regarding the view party protocol, if an honest party already has a commit proof, the only message it can send is the commit suggestion, which it does at most once per party and only if receiving a $\textsc{RequestSuggestion}$ from the leader, which implies that the leader is byzantine. Therefore, there is at most $f \cdot (n - f) = \mathcal{O}(n\cdot f)$ commit suggestion messages sent after GST. These are the only two messages that can be sent by honest parties when all of them have decided, therefore the message complexity for this part is $\mathcal{O}(n\cdot f)$.
 Taking the messages in these two parts into account, the total message complexity for the protocol is thus $\mathcal{O}(n\cdot f)$.
\end{proof}

\ViewBaMain*

\begin{proof}
    This is a direct consequence of \cref{lemma:gst-complexity,coro:gst-f-views,lemma:gst-proof-good,lemma:gst-validity}.
\end{proof}

\subsection{Complete synchronous protocol}\label{sec:sync-proto-ba}

We now give the pseudocode for our optimal byzantine agreement \textsc{SyncBA} in the synchronous setting with message complexity $\mathcal{O}(nf)$, round complexity $\mathcal{O}(f)$ and resiliency $t < n/2$:

\begin{dianabox}{\textsc{SyncBA}}
\algoHead{Optimal synchronous agreement protocol}
\begin{algorithmic}[1]
    \State $fallback\_proof \gets \bot$
    \For{$n$ views}
        \State Run $\textsc{ViewByzantineAgreement}$
    \EndFor
    \State Let $(value, commit\_proof)$ the commit value and proof \Statex \hspace{\algorithmicindent} obtained by $\textsc{ViewByzantineAgreement}$
    \State Let $\textbf{lock}$ the lock value obtained by $\textsc{ViewByzantineAgreement}$
    \Statex
    
    \Statex \textbf{Round $11 n + 1$:}
    \If{$value = \bot$}
        \State Broadcast $(HELP, tsign(HELP))$ to every party
    \EndIf
    
    \Statex
    \Statex \textbf{Round $11 n + 2$:}
    \UponSimple{Receiving $(HELP, tsign(HELP))$ from party $p$}
        \If{$value \ne \bot$}
            \State Send $(PROOF, value, commit\_proof)$ to $p$
        \EndIf
    \EndUpon
    \UponTrue{Receiving $(HELP, tsign(HELP))$ from $t+1$ different parties}
        \State Make $fallback\_proof$ a $(t+1)$-threshold certificate from the partial signatures
        \State Broadcast $(FALLBACK, fallback\_proof)$ to every party
    \EndUpon
    
    \Statex
    \Statex \textbf{Round $11 n + 3$:}
        \UponSimple{Receiving valid $(PROOF, v, proof)$}
            \State $value \gets v$
            \State $commit\_proof \gets proof$
        \EndUpon
        \UponTrue{Receiving valid $(FALLBACK, proof)$}
            \State $fallback\_proof \gets proof$
            \If{$lock \ne \bot$}
                \State Broadcast $(LOCK, lock)$ to every party
            \EndIf
        \EndUpon
        
    \Statex
    \Statex \textbf{Round $11 n + 4$ and followings:}
        \If{$value = \bot$}
            \If{Received a valid $(LOCK, lock)$}
                \State Let $v$ be the lock value received with the highest view number
                \State $value \gets v$
            \Else 
                \State $value \gets v_{in}$
            \EndIf
        \EndIf
        \If{$fallback\_proof \ne \bot$}
            \State $v \gets \mathcal{A}_{fallback}(value)$
            \If{$commit\_proof = \bot$}
                \State $value \gets v$
            \EndIf
        \EndIf
        \If{The party has not yet decided}
            \State Decide $value$
        \EndIf
\end{algorithmic}
\end{dianabox}

\begin{lemma}\label{lemma:gst-no-fallback}
    If a fallback proof is created, then $f \geq n-k$, and in particular $f = \Omega(n)$.
\end{lemma}

\begin{proof}
    If a fallback proof is created, it implies that $t+1$ parties did not get a commit proof during $n$ views, which implies that one honest party did not decide during $n \geq 2t+1\geq 2f+1$ views.

    However, using \cref{lemma:gst-decide-2f} and the fact that we are in synchrony, so the $GST$ is set at the beginning of the algorithm, one honest party not deciding with $2f+1$ views imply that $f \geq n-k = n - \lceil \frac{n+t+1}{2} \rceil \approx n/4$. Therefore, $f = \Omega(n)$.
\end{proof}

\begin{theorem}\label{theo:sync-complex}
    The protocol has $\mathcal{O}(f)$ round complexity and $\mathcal{O}(n\cdot f)$ message complexity.
\end{theorem}

\begin{proof}
    We look at the two cases ($f < n-k$ and $f \geq n-k$) separately:
    
    \textbf{If $f < n-k \approx n/4$}: then using \cref{lemma:gst-complexity}, the view-baed protocol has $\mathcal{O}(n\cdot f)$ message complexity. Moreover, using \cref{lemma:gst-honest-commit}, because $n \geq 2t+1\geq 2f+1$, all honest parties will have decided within $\mathcal{O}(f)$ rounds. Moreover, it implies that all honest parties will have a proof by the end of the $n$ views. Therefore, no honest party will send a $HELP$ message. Honest parties will answer once to each $HELP$ message, which can thus only be sent the $f$ byzantine parties. This has message complexity $(n-f)\cdot f$. We remark that any other message honest parties can send is behind a fallback proof check or receiving $t+1$ partial signatures for help. Because no honest party send a partial signature for help, the later case will not happen. And because of \cref{lemma:gst-no-fallback}, the former cannot happen either. Therefore, the total message complexity is $\mathcal{O}(n \cdot f)$.
    
    \textbf{If $f \geq n-k \approx n/4$}: then $f = \Omega(n)$. We remark using \cref{coro:gst-bad-complexity} that this part has message complexity $\mathcal{O}(n^2) = \mathcal{O}(n \cdot f)$. Moreover, this runs in $\mathcal{O}(n)$ rounds. In the following rounds, each honest party will do at most $3$ broadcasts and reply to help messages at most once per party. This has message complexity $4n^2 = \mathcal{O}(n f)$. The only remaining step is to run the fallback protocol. We note that it is not guaranteed that all honest parties will run it but \cref{theo:momose-good} guarantees that it will have $\mathcal{O}(n)$ round complexity and $\mathcal{O}(n^2) = \mathcal{O}(nf)$ message complexity anyway. So the total message complexity is $\mathcal{O}(nf)$. Moreover, the view-based approach takes $\mathcal{O}(n)$ rounds, then checking for fallback takes $3$ rounds and finally the fallback takes $\mathcal{O}(n)$ rounds, after which all parties will decide if they have not yet. Therefore, the total round complexity is $\mathcal{O}(n)$.
\end{proof}

\begin{theorem}\label{theo:sync-agree-val}
    This protocol satisfies agreement and validity.
\end{theorem}

\begin{proof}
    We look at different cases depending on whether a commit proof was acquired or not. We note that the adversary has a lot of power in this protocol. It can for example get a commit proof or fallback proof, but keep it for himself on only selectively show it later on. Nevertheless, we prove that we always achieve agreement and validity.

    \textbf{If at least one honest party has a commit proof by the end of the view-based protocol:} Then this party will send its commit value to any other party which didn't have it and sent a help request. Therefore, all honest parties will get a commit proof which can only exists for one value, as proven in \cref{lemma:gst-agreement}. They will also all ignore the output of the fallback algorithm if they run it and decide this same value which satisfies validity because of \cref{lemma:gst-validity}. So we have both validity and agreement.

    We can now assume that no honest party had a commit proof right at the end of the view-based part. So they will all send a help partial signature. Because there are at least $t+1$ honest parties, each honest party will be able to combine the messages into a fallback proof. We note that this implies that all honest parties will run the fallback protocol. We now have two new cases to look at:

    \textbf{If at any point an honest party got a commit proof}: We note that this is possible if the adversary creates a commit proof but only shares it after honest parties send their help message. A commit proof contains a view number. Let $view1$ be the view number of one such commit proof for value $v$. By definition, this proof is made from the partial signatures of $k$ parties who got a lock proof for $v$ at view $view1$. So at least $k-t$ of these parties were honest. We can now use \cref{lemma:klock-done} which guarantees that every subsequent lock proof created will be for $v$. We remark that $k - t = \lceil \frac{n+t+1}{2} \rceil - t > 0$, therefore there must be an honest party $p$ which got a lock proof for $v$ at round $view1$. Because all subsequent lock proofs will also be for $v$ and $p$ sends its lock proof to everyone, it means that all honest party who did not get a commit proof will set their value to $v$ as the lock proof with the highest view number they receive will be for it because they received it from at least $p$. Therefore, all honest parties will join the fallback protocol with input either the commit value $v$ or the highest lock value which is also $v$. Because it satisfies validity, its output will also be $v$. Thus, whether they decide the output of the fallback protocol or their commit value, they will decide $v$ so we have agreement. Finally, $v$ has a commit proof so using \cref{lemma:gst-validity}, it satisfies validity.

    \textbf{If no honest party ever gets a commit proof}: Then all honest parties will decide on the output of the fallback protocol. Because it satisfies agreement, the decided value will also satisfy agreement. Finally we remark that honest parties will join the fallback protocol with either their input value, or a lock value which according to \cref{lemma:gst-proof-good} satisfies validity. So using the fallback protocol validity, the output value also satisfies validity.
\end{proof}

\syncNf*
\begin{proof}
    \textsc{SyncBA} satisfies these properties by \Cref{theo:sync-complex} and \Cref{theo:sync-agree-val}.
\end{proof}

\subsection{Synchronous fallback protocol analysis}
We remark that in the synchronous setting, if $f \geq n/4$, we rely on the synchronous BA protocol by Momose and Ren \cite{momosequadraticba}. We note that this paper never covers the round complexity of their algorithm. For completeness, we prove that it is linear:

\begin{theorem}\label{theo:momose-good}
    The protocol for for byzantine agreement by Momose and Ren \cite{momosequadraticba} with optimal resiliency $t < n/2$ and quadratic communication complexity has $O(n)$ round complexity and $\mathcal{O}(n^2)$ message complexity. This is true even if not all honest parties join the protocol.
\end{theorem}

\begin{proof}
    We remark that their protocol consists of a recursive algorithm. This algorithm, when called on a set of parties of size $m$ larger than some constant $M$, consists of:
    \begin{itemize}
        \item Two recursive calls to the same protocol, with subsets of size $\lceil m / 2 \rceil$ then $\lfloor m /2 \rfloor$
        \item Two multicasts, each taking a single round
        \item Two calls to a graded byzantine agreement (GBA) procedure 
    \end{itemize}
    We remark that the GBA procedure they use takes exactly 4 rounds. Therefore, the total round complexity $T(m)$ for $m$ parties satisfies:
    \begin{align*}
        T(m) = \begin{cases}
            O(1), & \text{if } s \le M, \\
            T(\lfloor m/2 \rfloor) + T(\lceil m/2 \rceil) + O(1), & \text{otherwise.}
        \end{cases}
    \end{align*}
    For a given $m \geq 1$, the recursion depth $h$ satisfies $h \leq \lceil \log(m / M) \rceil \leq \log(m) + 1$. Therefore, we get:
    \begin{align*}
        T(m) &\leq \sum_{i=1}^h 2^{i} \cdot \mathcal{O}(1) \\
        & = (2^{h+1} - 1) \cdot \mathcal{O}(1) \\ 
        &\leq (4 \cdot 2^{\log m} - 1)\cdot \mathcal{O}(1) \\
        &= \mathcal{O}(m)
    \end{align*}
    As a consequence, the total round complexity of their protocol for $n$ parties is $\mathcal{O}(n)$. We also remark that almost the whole communication scheme (which party sends to which and when) is fixed in advance for this protocol. The only exception is in the graded consensus part where a party will perform between $1$ and $4$ multicasts depending on value received, but we remark that this only impact the constant factor in the message complexity, given that one multicast is always done. Therefore, even if not all honest parties join the protocol, the message complexity will still be $\mathcal{O}(n^2)$ while the round complexity will remain $\mathcal{O}(n)$.
\end{proof}

\section{QAB Phase in the partially synchronous setting}
\label{sec:apx:qab phase gst}

In this section, we provide a pseudocode for the QAB Phase procedure in the partially synchronous setting. To be more specific,  \textsc{QABRelayerGST} replaces \textsc{QABRelayer} while \textsc{QABQuorumGST} replaces \textsc{QABQuorum}. We note that the number of committees is $|\mathcal{C}| = \mathcal{O}(\hat{f} \log n)$.

\begin{dianabox}{\textsc{QABRelayerGST}}
\algoHead{Protocol for a relayer $r$ in the $i$-th committee $C \in  \mathcal{C}$ in partial synchrony}
\begin{algorithmic}[1]

    \State $Aggregate \gets \bot$
    \State $Signatures \gets \emptyset$
    \State $Seen \gets \emptyset$
    \UponTrue{Receiving valid $(DISPERSE, v, proof)$}
        \For{party $p \in C$}
            \State Send $(AGGREGATE, v, proof)$ to $p$
        \EndFor
    \EndUpon
    \Statex
    
    \UponSimple{Receiving valid $(ACK, sign)$ from a committee party $p$ for the first time}
        \State $Signatures \gets Signatures \cup \{sign\}$
        \If{$|Signatures| = |C|$}
            \State $Aggregate \gets aggregate(Signatures)$
            \State $Seen \gets \emptyset$
        \EndIf
    \EndUpon
    \Statex
    \For{view number $view$ in $0,1,\ldots$}
        \UponSimple{Receiving valid $(DISPERSE, v, proof)$ from quorum node $q \notin Seen$}
            \If{$view \equiv i \pmod{|\mathcal{C}|}$}
                \State $Seen \gets Seen \cup \{q\}$
                \State Wait $2\Delta$ units of time
                \If{$Aggregate \ne \bot$}
                    \State Send $(COMPLETED, Aggregate)$ to $q$
                \EndIf
            \EndIf
        \EndUpon
    \EndFor

\end{algorithmic}
\end{dianabox}

\begin{dianabox}{\textsc{QABQuorumGST}$(\vin, proof)$}
\algoHead{Protocol for a quorum node $q$ in partial synchrony}
\begin{algorithmic}[1]
    \State $Acknowledged \gets \emptyset$

    \InParallel{}
        \Loop{}
            \If{$Acknowledged = \mathcal{P}$}
                \State We are done disseminating our value 
                \State \Return
            \Else{}
                \State Let $p \in \mathcal{P} \setminus Acknowledged$
                \State Send $(DECIDE, v, proof)$ to $p$
                \State $Acknowledged \gets Acknowledged \cup \{p\}$
                \State Wait $\Delta$ units of time
            \EndIf
        \EndLoop
    \EndParallel
    \Statex
    
    \UponSimple{Receiving valid $(COMPLETED, sign)$ from relayer $r$ of committee $C$ for the first time}
        \State $Acknowledged \gets Acknowledged \cup C$
    \EndUpon
    \Statex

    \For{view number $view$ in $0,1,\ldots$}
        \State $i \gets view\pmod{|\mathcal{C}|}$
        \State Send $(DISPERSE, v, proof)$ to the relayer of the $i$-th committee
    \EndFor
    
\end{algorithmic}
\end{dianabox}

\end{document}